 \documentclass[lettersize,journal]{IEEEtran}

\IEEEoverridecommandlockouts
\usepackage{amsthm,amsmath,amssymb}
\usepackage{caption}
\usepackage{graphicx}
\usepackage{stfloats}
\usepackage[font=small]{caption}
\usepackage{cite}
\usepackage{booktabs}
\usepackage[labelformat=simple]{subcaption}
\usepackage{soul}
\usepackage{tabularx}
\usepackage{textcomp}
\usepackage{mathtools}
\usepackage{color}
\usepackage[dvipsnames]{xcolor}
\usepackage[normalem]{ulem}
\usepackage[thicklines]{cancel}

\newcommand{\mbf}{\mathbf}
\newcommand{\mrm}{\mathrm}
\newcommand{\bsm}{\boldsymbol}
\newcommand{\mal}{\mathcal}
\newcommand{\mbb}{\mathbb}
\newtheorem{definition}{Definition}
\newtheorem{assumption}{Assumption}
\newtheorem{proposition}{Proposition}

\newtheorem{remark}{Remark}

\usepackage{array}
\usepackage{algorithmic}
\usepackage{amsmath}

\usepackage{graphicx,float}
\usepackage[utf8]{inputenc}
\usepackage{enumerate}
\usepackage{color}
\usepackage{amsfonts}
\usepackage{algorithm}
\usepackage{algorithmic}
\def\BibTeX{{\rm B\kern-.05em{\sc i\kern-.025em b}\kern-.08em
    T\kern-.1667em\lower.7ex\hbox{E}\kern-.125emX}}

\begin{document}
%
% paper title
% Titles are generally capitalized except for words such as a, an, and, as,
% at, but, by, for, in, nor, of, on, or, the, to and up, which are usually
% not capitalized unless they are the first or last word of the title.
% Linebreaks \\ can be used within to get better formatting as desired.
% Do not put math or special symbols in the title.
% \title{Simultaneous Active and Passive Information Transfer for Reconfigurable Intelligent Surface Aided MIMO Systems}
% \title{{Active$\&$Passive Information Transfer for RIS-Aided MIMO-OFDM Systems: Joint Channel Estimation and Information Recovery}}
\title{Interference-Cancellation-Based Channel Knowledge Map Construction and Its Applications to Channel Estimation}

\author{Wenjun Jiang,~\IEEEmembership{Graduate Student Member,~IEEE}, 
	Xiaojun Yuan,~\IEEEmembership{Senior Member,~IEEE}, Boyu Teng,~\IEEEmembership{Graduate Student Member,~IEEE}, Hao Wang,~\IEEEmembership{Member,~IEEE}, and Jing Qian,~\IEEEmembership{Member,~IEEE} 
	
	% \thanks{
	% 		% This work has been accepted for publication in IEEE Transactions on Signal Processing, 2023.

	% 		% This work was supported in part by General Program of National Natural Science Foundation of China under Grant 62071090 and in part bySichuanScienceandTechnologyProgram under Grant 2022ZYD0120. The preliminary version of this work has been accepted for publication in the proceedings of 2023 IEEE International Symposium on Information Theory (ISIT). \emph{(Corresponding author: Xiaojun Yuan.)}
						
	% 		W. Jiang and X. Yuan are with the National Key Laboratory of Wireless Communications, University of Electronic Science and Technology of China, Chengdu 611731, China. (e-mail: wjjiang@std.uestc.edu.cn;  xjyuan@uestc.edu.cn).

	% 		M. Di Renzo is with Université
	% 		Paris-Saclay, CNRS, CentraleSupélec, Laboratoire des Signaux et Systèmes, 3 Rue Joliot-Curie, 91192 Gif-sur-Yvette, France. (e-mail: marco di-renzo@universite-paris-saclay.fr)
	% 		}
	
	}

\maketitle

%\section{Abstract}
%%Reconfigurable intelligent surface (RIS) has been envisioned as a promising technology for sixth-generation network. 
%
%%Recent information-theoretic research on reconfigurable intelligent surface (RIS) suggested that applying RIS for information transfer significantly improves the multiplexing gain, compared to applying RIS for passive beamforming.
%
\vspace{-0.7 cm}
\begin{abstract}
Channel knowledge map (CKM) is viewed as a digital twin of wireless channels, providing location-specific channel knowledge for environment-aware communications. A fundamental problem in CKM-assisted communications is how to construct the CKM efficiently. Current research focuses on interpolating or predicting channel knowledge based on error-free channel knowledge from measured regions, ignoring the extraction of channel knowledge. This paper addresses this gap by unifying the extraction and representation of channel knowledge. We propose a novel CKM construction framework that leverages the received signals of the base station (BS) as online and low-cost data. Specifically, we partition the BS coverage area into spatial grids. The channel knowledge per grid is represented by a set of multi-path powers, delays, and angles, based on the principle of spatial consistency. In the extraction of these channel parameters, the challenges lie in strong inter-cell interferences and non-linear relationship between received signals and channel parameters. To address these issues, we formulate the problem of CKM construction into a problem of Bayesian inference, employing a block-sparsity prior model to characterize the path-loss differences of interferers. Under the Bayesian inference framework, we develop a hybrid message-passing algorithm for the interference-cancellation-based CKM construction. Based on the CKM, we obtain the joint frequency-space covariance of user channel and design a CKM-assisted Bayesian channel estimator. The computational complexity of the channel estimator is substantially reduced by exploiting the CKM-derived covariance structure. Numerical results show that the proposed CKM provides accurate channel parameters at low signal-to-interference-plus-noise ratio (SINR) and that the CKM-assisted channel estimator significantly outperforms state-of-the-art counterparts.
\end{abstract}


\vspace{-0.3cm}
\section{Introduction}

Sixth-generation (6G) wireless networks are envisioned to provide ubiquitous connectivity (up to $10^7$ nodes per $\mrm{km}^2$), ultra-high throughput (peak rate of Tbps), hyper-reliable and low-latency communications (HRLLC), and centimeter-level positioning accuracy around 2030 \cite{Zhengquan,you2021towards}. Achieving these key performance indicators (KPIs) necessitates the exploration of various cutting-edge technologies. One promising advancement is the evolution from massive multiple-input multiple-output (MIMO) to extremely large-scale MIMO (XL-MIMO) \cite{XLMIMO}, which unlocks ultra-high spatial multiplexing gains. In parallel, millimeter-wave and terahertz communications have been extensively studied to exploit the vast bandwidth for high data rates. The applications of XL-MIMO, millimeter-wave, and terahertz techniques lead to a substantial increase in node density, array size, and communication bandwidth. Consequently, the dimension of the channel response expands dramatically in both frequency and space domains. This expansion introduces significant challenges to physical-layer processing, including the increased computational complexity in channel state information (CSI) acquisition and the heavy training overhead required for effective beamforming. 

% On the other hand, dense wireless nodes and high-dimensional channels imply an abundance of wireless data which carry the characteristics of local propagation environment. This drives a paradigm shift towards environment-aware communications. Existing environment-aware technologies include physical environment maps \cite{Map} and radio environment maps \cite{zhao2006network}. A physical environment map comprises 3D models of scatterers, electromagnetic parameters, and propagation mechanisms of electromagnetic waves. Based on this information, a ray-tracing platform is established to simulate the deterministic propagation of electromagnetic waves \cite{Ray}. Due to intensive computation in ray tracing, physical environment maps are mainly employed for offline system evaluation and optimization. Radio environment maps primarily focus on spectrum utilization information and are widely applied in cognitive radio systems for interference management and resource allocation \cite{Yilmaz}.

On the other hand, dense wireless nodes and high-dimensional channels imply an abundance of wireless data that carries the characteristics of local propagation environment. This drives a paradigm shift towards environment-aware communications. A representative technology of environment awareness is the radio map \cite{RME} which captures the spatial distribution of radio signal characteristics. The construction of a radio map can be converted to characterizing the propagation of electromagnetic (EM) waves. To this end, ray-tracing is a promising technique by effectively solving Maxwell’s equations \cite{Ray}. However, the ray-tracing simulation involves intensive computation and relies on the 3D models and EM properties of scatters, limiting its application. 
%Recent research \cite{WINERT} has developed a neural network-based ray-tracer that models the reflection and direct propagation of EM waves with lower complexity. 
Besides ray-tracing-based radio maps, the authors in \cite{JianWen} exploit the positions and shapes of scatterers and propose a random forest method for a path-loss map. Typically, radio maps focus on power-related metrics such as path-loss \cite{JianWen}, interference power \cite{Int}, and power spectral density \cite{PSD}, but these metrics do not fully capture the intrinsic characteristics of wireless channels.

Recently, channel knowledge map (CKM) has been proposed as an advancement of radio maps, representing an initial version of the digital twin for wireless channels \cite{CKM2}. 
CKM can be expressed as a look-up table or mapping from a user location to channel knowledge, reflecting the intrinsic characteristics of local propagation environment. The input location can be effectively obtained by emerging integrated sensing and communication (ISAC) technologies \cite{ISAC} or existing positioning systems such as the GPS system. The output channel knowledge can be tailored to specific requirements, including scenario information (such as LOS existence) and channel model parameters (such as multi-path delays and angles). Furthermore, the running latency is substantially lower than that of ray-tracing simulator \cite{Ray}, enabling CKM-assisted real-time physical-layer processing. In \cite{BF_CKM}, the authors introduced a channel angle map (CAM) and a beam index map (BIM) to enable hybrid beamforming with light training overhead. In \cite{Lol_CKM}, the authors leveraged CKM to provide positioning anchors and optimize the selection of anchors for positioning enhancement. Besides CKM-assisted communications, existing CKM-related research includes the construction of CKM. The authors in \cite{CKM1} proposed spatial interpolation methods to construct the CKM across an entire target region using channel data from partially measured regions. In \cite{EM_CKM}, the authors adopted the expectation-maximization (EM) method to develop a channel gain map (CGM).

Current methods for CKM construction \cite{CKM1,EM_CKM} assume noise-free channel parameters from a subset of regions and aim to interpolate (or predict) channel parameters at unmeasured locations, ignoring the crucial step of parameter extraction. One way to acquire noise-free channel parameters is offline measurements with high transmission powers and high-precision channel sounders \cite{Chl_sounder}. This approach is costly and highly sensitive to changes in the propagation environment. To mitigate these problems, a promising alternative is to store and reuse received signals at the base station (BS) as low-cost and online data. However, the use of BS data induces new challenges. The transmission powers of user terminals are typically limited, and strong interferences from neighboring cells are common \cite{TS38213}. Moreover, the received signals are nonlinear functions of the channel parameters (e.g., the array response of a path angle as a complex exponential function \cite{tse2005fundamentals}). As a result, the extraction of channel parameters becomes critical to CKM accuracy and cannot be ignored. To address this gap, we focus on constructing CKM in uplink multiple-input multiple-output orthogonal frequency-division multiplexing (MIMO-OFDM) systems under inter-cell interference (ICI). We propose a comprehensive framework that unifies the extraction and representation of channel parameters using BS data. To enhance CKM accuracy, we develop an interference-cancellation-based message-passing algorithm under a Bayesian inference framework. Furthermore, we apply the CKM to solve the fundamental problem of channel estimation and design a high-performance channel estimator. The main contributions of this paper are listed as follows:
\vspace{-0.05cm}
\begin{itemize}
    \item \textbf{CKM Construction Framework Based on BS received signals with ICI}: We establish a CKM construction framework that leverages the received signals of BS in MIMO-OFDM systems with ICI. The BS coverage area is divided into spatial grids with the received signals divided accordingly. In each grid, we represent the channel knowledge  by a set of channel model parameters based on the principle of spatial consistency \cite{TR38901} (which indicates the high channel similarity within a correlation distance). The channel model parameters comprise the number of effective paths, path powers, delays, and azimuth/elevation angles of arrival, capturing all elements of the channel response except for small-scale fading coefficients. By adjusting the number of effective paths, the CKM strikes a balance between the channel representation error and the construction cost.
    \item \textbf{Interference-Cancellation-Based Bayesian CKM Construction}: We formulate the problem of CKM construction into a Bayesian inference problem, where interference channels are modeled as sparse matrices to characterize the significant differences of path-loss among interferers. By employing the Bethe method \cite{BP}, we construct a variational approximation of the posterior distribution of user/interference channels. Then, we develop a hybrid message-passing algorithm for the variational solution. In each message-passing iteration, the algorithm simultaneously extracts the user channel parameters and cancels the interference components.
    % extract the interference channel parameters for interference cancellation.
    \item \textbf{CKM-Assisted Bayesian Channel Estimator}: We design a Bayesian channel estimator that exploits the joint frequency-space covariance of the user channel provided by the CKM. Besides, the spatial covariance of interference signals is estimated based on the power delay profile (PDP) of the user channel from the CKM. By exploiting the structural property of the user channel covariance, we reduce the computational complexity of the channel estimation from $\mathcal{O}(N^3 M^3)$ to $\mathcal{O}(M^2(M+\bar{L})+(M+N)\bar{L}^2+\bar{L}^3)$ for an $N$-subcarrier and $M$-antenna system assisted by a $\bar{L}$-path CKM. Furthermore, we analyze the performance bound of the channel estimator in terms of the accuracy of CKM outputs.
\end{itemize}
Simulation results show that the CKM construction algorithm is able to extract accurate channel parameters at relatively low SINRs and meter-level spatial grids. As the grid size increases, the desired CKM accuracy can be maintained by increasing the number of effective paths. Besides, the CKM-assisted channel estimator significantly outperforms state-of-the-art schemes.

\vspace{-0.02cm}
{\emph{Organization:} In Sec. II, we introduce the received signal model for the considered MIMO-OFDM system with ICI. In Sec. III, we propose the framework of CKM construction based on the BS received signals and spatial consistency principle. In Sec. IV, we formulate the problem of CKM construction into a Bayesian inference problem and develop a hybrid message-passing algorithm for CKM construction. In Sec. V, we exploit the CKM outputs to design a high-performance channel estimator. In Sec. VI and Sec. VII, we provide numerical results and conclude this paper, respectively.}

\vspace{-0.02cm}
\emph{Notation:} We use bold capital letters (e.g., $\mathbf X$) for matrices and bold lowercase letters (e.g., $\mathbf x$) for vectors. $(\cdot)^T$, $(\cdot)^*$, and $(\cdot)^H$ denote the transpose, the conjugate, and the conjugate transpose, respectively. $\delta(\cdot)$ denotes the Dirac delta function. ${\rm diag}(\mathbf x)$ denotes the diagonal matrix with the $i$-th diagonal entry being the $i$-th element of $\mbf x$. $\mrm{tr}(\cdot)$ and $\mrm{vec}(\cdot)$ denote the trace and vectorization operators, respectively. $(\cdot)^{\dagger}$ denotes the Moore-Penrose inverse. $\otimes$ and $\odot$ denote the Kronecker and Hadamard product, respectively. Matrix $\mbf I$ denotes an identity matrix with an appropriate size. For a random variable $x$, its probability density function (pdf) is denoted by $p(x)$. $\mathbb E[\cdot]$ denotes the expectation operator. The pdf of a complex Gaussian random vector $\mathbf x \in \mathbb C^N$ with mean $\bsm{\mu}$ and covariance $\bsm{\Sigma}$ is denoted by $\mathcal{CN}(\mathbf x ; \bsm{\mu} , \bsm{\Sigma}) = |\bsm{\Sigma}|^{-1} {\rm exp}(-(\mathbf x - \bsm{\mu})^H (\bsm{\Sigma})^{-1} (\mathbf x - \bsm{\mu}))/\pi^N$. The pdf of a von Mises (V-M) random variable $ x \in [0,2\pi)$ with mean direction $\mu$ and concentration $\kappa$ is denoted by $ \mal{VM}(x;\mu,\kappa)=(2 \pi I_0(\kappa))^{-1}\exp(\kappa \cos(x-\mu))$, where $I_0(\cdot)$ denotes the modified Bessel function of the first kind and order $0$.

\vspace{-0.3cm}
\section{System Model}

\begin{figure}[h] 
	\centering
	\includegraphics[width = 3.5 in]{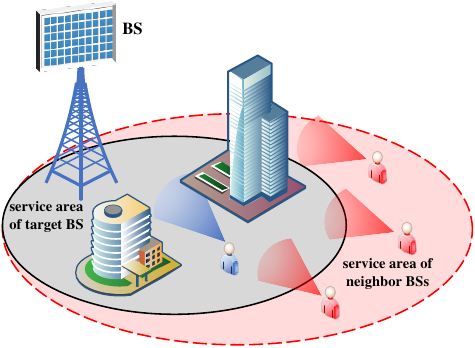}
	\caption{An illustration of the uplink transmission system with a target user and $3$ inter-cell interferers.}
	\label{system}
\end{figure}

Consider the uplink MIMO-OFDM system illustrated in Fig. \ref{system}. The BS is equipped with a uniform rectangular array (URA) of $M = M_1 \times M_2$ antennas, where the antenna spacing is assumed to be half-wavelength. Note that a BS can serve massive users by time-division multiplexing (TDD) and frequency-division multiplexing (FDD). Since the signal model for each user has the same form, we focus on a target user that communicates with the BS over $N$ subcarriers in the $t$-th time-slot. Denote by $\Delta_f$ the subcarrier spacing. The target user is equipped with an omnidirectional antenna, and the user signal experiences a multi-path channel. Denote by $L_t$ the number of multi-paths in time-slot $t$. Denote by $\tilde{\tau}_{t,l}$, $\tilde{\theta}_{t,l}$, and $\tilde{\phi}_{t,l}$ the path delay, the azimuth angle-of-arrival (AOA), and the zenith AOA of the $l$-th path in the $t$-th time-slot, respectively. Let $\tau_{t,l} \triangleq 2 \pi \Delta_f \tilde{\tau}_{t,l}$, ${\theta}_{t,l} \triangleq \pi \sin \tilde{\theta}_{t,l} \cos \tilde{\phi}_{t,l} $, and $\phi_{t,l} \triangleq \pi \sin \tilde{\phi}_{t,l} $ be the normalized path delay, the directional component in the azimuth direction, and the directional component in the zenith direction, respectively. Then, the frequency-space-domain channel of the target user at time-slot $t$ is given by \cite{tse2005fundamentals}
\begin{align}
	\mbf H_{t} = \sum_{l=1}^{L_t} \alpha_{t,l} \sqrt{\rho_{t,l}} \mbf a_N(\tau_{t,l}) \left(  \mbf a_{M_1}(\theta_{t,l}) \otimes \mbf a_{M_2}(\phi_{t,l}) \right)^H, \label{H_t}
\end{align}
where $\rho_{t,l}$ represents the $l$-th path power; $\alpha_{t,l},\forall l$ represent the i.i.d. complex random coefficients with zero means and unit powers; $\mbf a_{x}(\omega), \omega \in \{\tau_{t,l},\theta_{t,l},\phi_{t,l} \}$ represents the steering vector as 
\begin{align}
	\mbf a_{x}(\omega) \triangleq  \frac{1}{\sqrt{x}} [ 1, e^{-i \omega }, ... , e^{-i  (x-1)  \omega } ]^T, \notag
\end{align}
where $i^2=-1$. In 5g cellular networks, neighbor cells share the frequency spectrum to improve network throughput \cite{TS38213}. On the other hand, the ICI becomes the bottleneck of high spectrum efficiency. In the $t$-th time-slot, denote by $k$ the interference index with $k \in \mal{K}_t = \{1,2,....,K_t\}$. Similar to \eqref{H_t}, the channel model of the $k$-th interferer is given by  
\begin{align}
	\mbf H_{t,k}^I = \sum_{l=1}^{L_k^I} \alpha_{t,k,l}^I \sqrt{\rho_{t,k,l}^I} & \mbf a_N(\tau_{t,k,l}^I) \times \notag \\
	&  \left(  \mbf a_{M_1}(\theta_{t,k,l}^I) \otimes \mbf a_{M_2}(\phi_{t,k,l}^I) \right)^H, \label{H_I}
\end{align}
where $\alpha_{t,k,l}^I$, $\rho_{t,k,l}^I$, $\tau_{t,k,l}^I$, $\theta_{t,k,l}^I$, and $\phi_{t,k,l}^I$ represent the $l$-th complex coefficient, path power, normalized delay, directional components in the azimuth and zenith directions, respectively. Denote by $\mbf X_t \in \mbb{C}^{N \times N}$ and $\mbf X_{t,k}^I \in \mbb{C}^{N \times N}$ the diagonal pilot matrices of the user and the $k$-th interferer, respectively. The diagonal elements of $\mbf X_t$ and $\mbf X_{t,k}^I$ are denoted by $\mbf x_{t}$ and $\mbf x_{t,k}^I$, respectively. $\mbf x_{t}$ and $\mbf x_{t,k}^I$ are adopted as constant envelope signals with powers being $P$ and $P^I$, respectively. Assume that the length of the cyclic prefix (CP) is greater than the maximum channel delay spread. After removing the CP and applying the discrete Fourier transform (DFT), the frequency-space-domain received signal is expressed as 
\begin{align}
	& \mbf Y_t = \mbf X_t \mbf H_t + \sum_{k}  \mbf X_{t,k}^I \mbf H_{t,k}^I + \mbf N_t, \label{Y_t}
\end{align}
or equivalently
\begin{align}
	& \mbf y_t = \sum_l \alpha_{t,l} \sqrt{\rho_{t,l}} \mbf b({\theta}_{t,l},{\phi}_{t,l})\otimes{\mbf s}_{t} ({\tau}_{t,l})  + \notag \\
	& \hspace{20pt} \sum_{k,l} \alpha_{t,k,l}^I \sqrt{\rho_{t,k,l}^I} \mbf b(\theta_{t,k,l}^I,\phi_{t,k,l}^I)\otimes{\mbf s}_{t,k}^I ({\tau}_{t,k,l}^I) + \mbf n_t, \label{Y_t2}
\end{align}
where $\mbf N_t$ is an additive white Gaussian noise (AWGN) matrix with the elements independently drawn from $\mal{CN}(0,\gamma)$. \eqref{Y_t2} is obtained by substituting \eqref{H_t} and \eqref{H_I} into \eqref{Y_t} and vectorizing $\mbf Y_t$. Specifically, \eqref{Y_t} and \eqref{Y_t2} are associated by noting $\mbf y_t = \mrm{vec}(\mbf Y_t)$, $\mbf n_t = \mrm{vec}(\mbf N_t)$, ${\mbf s}_{t}({\tau}_{t,l}) = \mbf x_{t} \odot \mbf a_N({\tau}_{t,l})$, $\mbf b(\theta_{t,l},\phi_{t,l}) =   \mbf a_{M_1}(\theta_{t,l}) \otimes \mbf a_{M_2}(\phi_{t,l}) $, ${\mbf s}_{t,k}^I ({\tau}_{t,k,l}) = \mbf x_{t,k}^I \odot \mbf a_N({\tau}_{t,l})$, and $\mbf b({\theta}_{t,k,l}^I,{\phi}_{t,k,l}^I) = \mbf a_{M_1}(\theta_{t,k,l}^I) \otimes \mbf a_{M_2}(\phi_{t,k,l}^I)$. In the remainder of this paper, we use model \eqref{Y_t} for channel estimation and use model \eqref{Y_t2} for CKM construction. To enhance the channel estimation performance, it is crucial to exploit the prior knowledge of $\mbf H_t$. To this end, we propose a CKM to provide strong priors, and the construction of CKM is described in Sec. III-IV. The interference cancellation is also essential for ensuring estimation performance, especially at a relatively low SINR. In Sec. V, we use the CKM outputs to assist the estimation of interference spatial covariance.

% Signal model under interference
% \begin{align}
% 	\mbf Y_t = \sum_l \alpha_{l,t} \tilde{\mbf a}_N(\tau_l) \left( \mbf W \left(  \mbf a_{M_1}(\theta_{1,l}) \otimes \mbf a_{M_2}(\theta_{2,l})  \right)  \right)^T + \sum_k \mbf X_k \mbf F_f \mbf S_k \mbf F_s^T \mbf W^T +  \mbf N_t.
% \end{align}
% We aim to extract $\tau_l$, $\theta_{1,l}$, $\theta_{2,l}$, and the average power of $\{\alpha_{l,t}\}_{l=1}^{L}$. We are agnostic to the model parameters of the interference channel and are primarily concerned with effectively mitigating strong interference components. Thus, we employ a basis transformation matrix under uniform sampling, with the aim of attenuating high-amplitude elements in the transformed domain.

% or
% \begin{align}
% 	\mbf Y = \sum_l (\mbf x_l \odot \mbf a_N(\tau_l)) \otimes \left( \mbf W \left(  \mbf a_{M_1}(\theta_{1,l}) \otimes \mbf a_{M_2}(\theta_{2,l})  \right)  \right) \bsm{\alpha}^T + \mbf N
% \end{align}

\section{Framework of Proposed CKM}

\subsection{Basic Assumptions}
With the rapid development of ISAC, user localization can be efficiently achieved in communicates systems. Against this background, the CKM is proposed as a mapping from the user location to some channel statistics during the transmission time-slots of interest\cite{CKM_tutorial}. A common assumption in CKM is that the propagation environment is quasi-static, that is, the propagation environment remains unchanged in a relatively large timescale compared to the period of signal transmission. To see the reasonableness of this assumption, an example is that the transmission interval of an OFDM symbol is $33.3$ $\mu$s under $\Delta_f = 30 $ KHz \cite{TS38211}, while the variations in propagation environment, e.g., building deformations, occur over the span of days. Also, note that the BS is typically static after deployment. Given the static BS and quasi-static propagation environment, the channel statistics can be regarded as a vector-valued function of the user localization.

\begin{figure}[h] 
    \centering
    \includegraphics[width = 3.5 in]{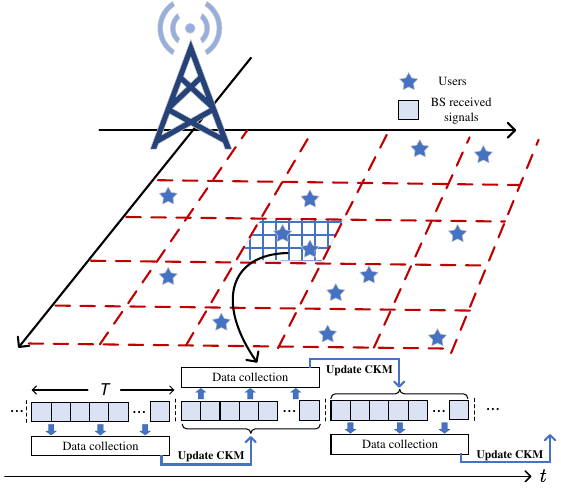}
    \caption{An illustration of the CKM construction framework.}
    \label{CKM_ill}
\end{figure}

Consider that the physical region is divided into $Q$ grids, and the CKM is periodically updated with period $T$. The process of updating the CKM per period follows the same procedure. Without loss of generality, we introduce the framework of CKM construction in one period. Define a mapping $\mal{L}(\cdot)$ from the time-slot $t$ to a location $q$ as
\begin{align}
 \mal{L}(\cdot): \mbb{T} \rightarrow \mbb{Q}, \quad q = \mal{L}(t), \label{q_t}
\end{align}
where $\mbb{T}=\{1,..., T\}$ and $\mbb{Q}=\{1,..., Q\}$ denote the set of the time-slots and the locations, respectively. The location $q$ corresponds to a physical grid of $d \times d ~ \mrm{m}^2$. Note that $\mal{L}(\cdot)$ is a many-to-one mapping, i.e., multiple time-slots associate with a location $q$. For the CKM construction, some assumptions are made below:
\begin{assumption}
 The CKM is constructed based on the following assumptions: 
    \begin{itemize}
        \item[(a)] At any location $q$, the changes of channel parameters $\rho_{t,l}$, $\tau_{t,l}$, $\theta_{t,l}$ and $\phi_{t,l}$ over $2T$ continuous time-slots are negligible.
        \item[(b)] $\mal{L}(\cdot)$ is obtained by the BS per time-slot.
        \item[(c)] Interference pilots $\{\mbf X_{k,t}^I\}_{\forall k,\forall t}$ are known by the BS at the beginning of the CKM update per period.
    \end{itemize}
\end{assumption} 
Assumption 1(a) implies that the quasi-static characteristics of the propagation environment hold within $2T$ time-slots. In practice, the value of $T$ should be appropriately set according to the rate of changes in the propagation environment (e.g., the changes in building shape, weather, etc). Assumption 1(b) can be satisfied by using various ISAC techniques or through cooperation with other positioning systems (e.g., the GPS system). As for Assumption 1(c), the neighbor-cell BSs can store the pilot usage information $\mbf X_{k,t}^I$ and share this information with the target BS. The signaling overhead from this information sharing is marginal since the update period $T$ is relatively long. Besides, we can adopt pseudo-random sequences as pilots, where only random seeds need to be shared.

% Furthermore, $T$ can be adaptively adjusted to cope with some abrupt changes in the propagation environment. For example, when the weather transitions from cloudy to rainy, it is necessary to use a short $T$ during the transition phase. 

We provide a diagram of the CKM construction in Fig. \ref{CKM_ill}.
% we associate $\mbf y_t$ in \eqref{Y_t2} with a location $q$.
For ease of description, define $\mbb{T}_q = \{ t |  q=\mal{L}(t) \}$ as the set of time-slots corresponding to location $q$. Within  $T$ time-slots, the received signals $\mbf y_t,t\in \mbb{T}_q$ are collected as measurement data, where the measurement data can be reference signals, synchronization signals, pilot signals, etc from all users at the $q$-th location. Then, $\mbf y_t,t\in \mbb{T}_q$ is used to construct the CKM, and the CKM outputs are exploited to assist the wireless communications in the following $T$ time-slots. 
% Recall that a BS can serve hundreds of users at various locations by TDD. Therefore, we can employ the historical data from not only the target user itself bus also other users in the same location $q$. 

\vspace{-0.3cm}

\subsection{Proposed CKM}
We now provide a specific description of the proposed CKM. From \eqref{H_t}, $\mbf H_t, \forall t \in \mbb{T}_q$ is determined by the parameters $\alpha_{l,t}$, $\rho_{l,t}$, $\tau_{l,t}$, $\theta_{l,t}$, and $\phi_{l,t}$. Among these parameters, $\alpha_{l,t}$ is a small-scale fading coefficient with the phase variation periodic to the wavelength, which is difficult to characterize. $\rho_{l,t}$, $\tau_{l,t}$, $\theta_{l,t}$, and $\phi_{l,t}$ are large-scale fading parameters and exhibit spatial consistency \cite{TR38901}. For example, consider a user-BS distance of $300$ meters. When the user moves $5$ meters perpendicular to the user-BS line, the delay and angle drifts of the line-of-sight (LOS) path are $0.14$ ns and $0.19$ degrees, respectively. This example indicates that large-scale fading parameters approximately remain constants in a physical region with a relatively small $d$. Based on the spatial consistency, we replace $\{\rho_{l,t},\tau_{l,t},\theta_{l,t},\phi_{l,t}\}_{l=1}^L$ by location-based parameter set as $\{ \bar{\tau}_{q,l},\bar{\theta}_{q,l},\bar{\phi}_{q,l},\bar{\rho}_{q,l} \}_{l=1}^{\bar{L}}$, where $\bar{L}$ represents the number of effective channel paths. The corresponding CKM is expressed as 
\begin{align}
 \mal{C}(\cdot) : q \rightarrow \{ \bar{\tau}_{q,l},\bar{\theta}_{q,l},\bar{\phi}_{q,l},\bar{\rho}_{q,l} \}_{l=1}^{\bar{L}}. \label{CKM}
\end{align}
Given \eqref{CKM}, the user channel $\mbf H_t,t \in \mbb{T}_q$ in \eqref{H_t} is represented as
\begin{align}
 \mbf H_t = \sum_{l=1}^{\bar{L}} \bar{\alpha}_{t,l} \sqrt{\bar{\rho}_{q,l}} \mbf a_N(\bar{\tau}_{q,l}) \mbf b(\bar{\theta}_{q,l},\bar{\phi}_{q,l})^H + \bsm{\Delta}_{\mbf H_t}, \label{apr_H_t}
\end{align} where $\bar{\alpha}_{t,l}$ is the $l$-th effective complex coefficient and $\bsm{\Delta}_{\mbf H_t}$ is the representation error. The value of $\bar{L}$ controls the trade-off between the representation error and the construction complexity. Note that $\bar{L}$ can be less than $L$ since an effective path can approximate two channel paths with similar delays and angles.  In \eqref{apr_H_t}, $\bar{\rho}_{q,l}$ is not only related to location $q$ but also related to the radiation pattern of the transceiver array \cite{balanis2016antenna}. Recall that we assume an omnidirectional antenna at the user, and the CKM is constructed specifically for each BS. Therefore, the impacts of radiation patterns on the user and the BS can be ignored in $\mal{C}(\cdot)$. As for path delays, they include the timing offset as a group delay added to $\bar{\tau}_{q,l},\forall l$. This group delay can be accurately calibrated \cite{Synchr}, and we focus on the extraction of path delays which are associated with propagation distances only.

%\begin{subequations}
%	\begin{align}
		% &  \bar{v}_{q,l} = v_{t,l} + \Delta_{v_{t,l}}, ~ \forall t \in \mbb{T}_q, 
		% \\
		% &  \bar{\theta}_{q,l} = \theta_{t,l} + \Delta_{\theta_{t,l}}, ~ \forall t \in \mbb{T}_q, \\ 
		% &  \bar{\phi}_{q,l} = \phi_{t,l} + \Delta_{\phi_{t,l}}, ~ \forall t \in \mbb{T}_q, \\ 
		% &  \bar{\rho}_{q,l} = \rho_{t,l} + \Delta_{\rho_{t,l}}, ~ \forall t \in \mbb{T}_q,
		
%	\end{align}	
%\end{subequations}
% $v \in \{\tau,\theta,\phi,\rho\}$ and 
%$\Delta_{v}$ is approximation error. 

With \eqref{apr_H_t}, we re-write \eqref{Y_t2} as a CKM-based signal model: 
\begin{align}
	\mbf y_t = & \sum_l \beta_{t,l} \mbf b(\bar{\theta}_{q,l},\bar{\phi}_{q,l})\otimes{\mbf s}_{t} (\bar{\tau}_{q,l}) \notag \\ 
	& + \sum_{k,l} \beta_{t,k,l}^I \mbf b(\theta_{t,k,l}^I,\phi_{t,k,l}^I)\otimes{\mbf s}_{t,k}^I ({\tau}_{t,k,l}^I) + \tilde{\mbf n}_t, \label{y_CKM}
\end{align} 
where $\beta_{t,l} = \bar{\alpha}_{t,l} \sqrt{\bar{\rho}_{q,l}} $; $\beta_{t,k,l}^I = \alpha_{t,k,l}^I \sqrt{{\rho}_{t,k,l}^I} $; $\tilde{\mbf n}_t$ is the addition of the AWGN and the representation error induced by $\bsm{\Delta}_{\mbf H_t}$. Note that $ \mbf y_t, \forall t \in \mbb{T}_q$ in \eqref{y_CKM} corresponds to the $q$-th location of the target user, and the second term in \eqref{y_CKM} consists of the signals from various interferes. For a performance-complexity balance in interference cancellation, we can cancel $\bar{L}_k^I < L_k^I$ path components which dominate the channel power of interferer $k$. In Sec. IV, we describe how to establish the CKM. Then, in Sec. V, we use the CKM outputs to assist a high-accuracy channel estimation. 
% We further show in Sec. IV that the CKM outputs can also assist the estimation of interference prior for interference cancellation. 

\begin{remark}
	The proposed CKM outputs $\{ \bar{\tau}_{q,l},\bar{\theta}_{q,l},\bar{\phi}_{q,l},\bar{\rho}_{q,l} \}_{l=1}^{\bar{L}} $ as location-specific priors which can be exploited for high-accuracy channel estimation. For comparison, the compressive-sensing (CS) methods \cite{OMP,AMP_2009,turbo_ma,STCS} adopt a dictionary which is typically constructed from the grid sampling of  $\mbf a_x(\omega),\omega \in\{\tau_{t,l},\theta_{t,l},\phi_{t,l}\}$. Based on this dictionary, the CS methods aim to find a sparse representation of the channel. However, a pre-determined dictionary induces the energy leakage due to the model mismatch between the discrete grid samples of $\mbf a_x(\omega)$ and the true $\mbf a_x(\omega)$ with continuous $\omega$. In CKM, the extraction of delays/angles is dictionary-free and regarded as a spectral estimation problem based on \eqref{y_CKM}. As for the neural-network (NN) methods \cite{Jinshi,Vincent,chenchen}, $\mbf H_t$ is treated as an image for correlation learning. To enhance generalization, various channel data (e.g., highly and weakly frequency-selective channels) are used to learn a common NN, potentially leading to a performance loss. In CKM, the output channel parameters reflect the channel characteristics in the specific propagation environment.
	% From the point of methodology, the CKM enables a channel estimation with higher accuracy, since its sparse prior or correlation prior is user-specific.
	
\end{remark}

\begin{remark}
	An innovation of the proposed CKM, compared to \cite{CKM1,EM_CKM}, is the consideration of practical interferences and the mechanism for interference cancellation. In practice, the interference strengths are unknown and vary due to dramatically different transmitting powers and path-loss. Then, for some interferences with severe path loss and low transmitting powers, the corresponding received signals are negligible. Consequently, we employ block sparsity to characterize the interference channels, as detailed in the following section. 
\end{remark}

\vspace{-0.3cm}

\section{Interference-Cancellation-Based Bayesian CKM Construction}

In Sec. IV-A, we establish a Bayesian inference framework for the CKM construction. To solve the Bayesian inference, we develop a message-passing principle in Sec. IV-B. Sec. IV-C provides the detailed derivations of message-passings. During message-passings, the parameters of prior probabilities are learned by expectation maximization (EM), as shown in Sec. IV-D. Sec. IV-E concludes the proposed algorithm. 
\vspace{-0.3cm}

\subsection{Problem Formulation}

From the data $\mbf y_t, t \in \mbb{T}_q$, we aim to cancel the interference signals and simultaneously extract $\bar{\tau}_{q,l}$, $\bar{\theta}_{q,l}$, $\bar{\phi}_{q,l}$, and $\bar{\rho}_{q,l}$ for the CKM construction\footnote{To reduce computational complexity, we can use a subset of $\{\mbf y_t| t \in \mbb{T}_q\}$ to construct the CKM, where the construction process is the same.}. To this end, we formulate the problem of CKM construction into a Bayesian inference problem. Note that the process of CKM construction per location $q$ is separate and follows the same procedure. Without loss of generality, we provide the problem formulation in the $q$-th  location below.  We first probabilistically model the channel parameters of the target user as
\begin{subequations}
    \label{p_usr}
    \begin{align}
        & \bar{v}_{q,l}  \sim \mal{VM}( \bar{v}_{q,l} ; \mu_{\bar{v}_{q,l}}^{\mrm{pri}}, \kappa_{\bar{v}_{q,l}}^{\mrm{pri}} ), \\
        & \beta_{t,l} \sim \mal{CN}( \beta_{t,l} ; 0, {\bar{\rho}}_{q,l} ), \label{p_rho}
    \end{align}
\end{subequations}
where $v \in \{\tau,\theta,\phi\}$; superscript ``$\mrm{pri}$'' is the abbreviation of ``prior''. In \eqref{p_rho}, we treat the path power ${\bar{\rho}}_{q,l}$ as the variance of $\beta_{t,l}$. The probability models of the interferences are given by
\begin{subequations}
    \label{p_I}
    \begin{align}
        & {v}_{t,k,l}^I  \sim \mal{VM}( {v}_{t,k,l}^I ; \mu_{{v}_{t,k,l}^I}^{\mrm{pri}}, \kappa_{{v}_{t,k,l}^I}^{\mrm{pri}} ), \\
        & \bsm{\beta}_{t,k}^I \sim (1-\lambda)\delta(\bsm{\beta}_{t,k}^I) + \lambda \mal{CN}(\bsm{\beta}_{t,k}^I; \mbf 0 , \mrm{diag}(\bsm{\rho}_{t,k}^I) ), \label{p_rho_I}
    \end{align}
\end{subequations}
where $v \in \{\tau,\theta,\phi\}$, $\bsm{\beta}_{t,k}^I = [\beta_{t,k,1}^I,...,\beta_{t,k,L_k^I}^I]^T$ and $\bsm{\rho}_{t,k}^I = [\rho_{t,k,1},...,\rho_{t,k,L_k^I}]^T$. In \eqref{p_rho_I}, the Bernoulli Gaussian pdf models that $\bsm{\beta}_{t,k}^I$ is a zero-vector with probability $1-\lambda$, which is used to characterize the potential severe path-loss of interferer $k$. Let $\bsm{\xi}_q$ be a vector of all elements of $\{ \bar{\tau}_{q,l},\bar{\theta}_{q,l},\bar{\phi}_{q,l},\beta_{t,l} \}_{t \in \mbb T_q,\forall l}$ with an arbitrary order. Similarly, let $\bsm{\xi}_t^I$ be a vector of all elements of $\{ {\tau}_{t,k,l}^I,{\theta}_{t,k,l}^I,{\phi}_{t,k,l}^I,\beta_{t,k,l}^I \}_{\forall k,\forall l}$. Given the observations $\{\mbf y_t | t \in \mbb T_q\}$, the posterior distribution of $\bsm{\xi}_q$ and $\{\bsm{\xi}_{t}^I\}_{t \in \mbb T_q}$ is expressed as 
\begin{align}
    & p( \bsm{\xi}_q , \{\bsm{\xi}_t^I\}_{t \in \mbb T_q} | \{ \mbf y_t \}_{t \in \mbb T_q} )   \notag \\ 
    & \propto \prod_{t} p(\mbf y_t | \bsm{\xi}_q,\bsm{\xi}_t^I) \prod_l p(\bar{\tau}_{q,l}) p(\bar{\theta}_{q,l}) p(\bar{\phi}_{q,l}) \prod_{t,l} p(\beta_{t,l}) \notag \\
    & \hspace{10pt} \prod_{t,k,l} p({\tau}_{t,k,l}^I) p({\theta}_{t,k,l}^I) p({\phi}_{t,k,l}^I) \prod_{t,k} p(\bsm{\beta}_{t,k}^I). \label{p_post}
\end{align}
Based on \eqref{p_post}, the Bayesian optimal minimum mean square error (MMSE) estimator requires high-dimensional integral which is computationally intractable. In the following section, we develop a message-passing algorithm to obtain an approximate solution.

\vspace{-0.3cm}
\subsection{Message-Passing Principle}
\begin{figure}[h] 
	\centering
	\includegraphics[width = 3.5 in]{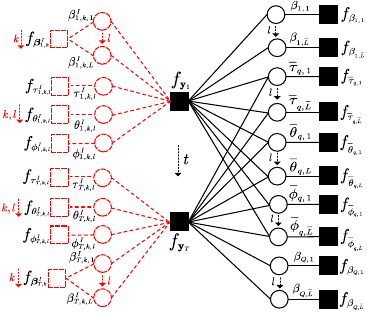}
	\caption{A factor graph representation of the factorization of posterior distribution \eqref{p_post}. $\{1,T\} \subseteq \mbb{T}_q$ in this example.}
	\label{FG}
\end{figure}

The Bethe method \cite{BP}, originated from nuclear physics, has been widely used to develop variational inference algorithms, known as \emph{message-passing}. The authors in \cite{Unify_1,Unify_2,HVMP} show that the Bethe framework under different constraints unifies three commonly used message-passing rules including belief propagation (BP) \cite{BP}, variational message passing (VMP) \cite{VMP}, and expectation propagation (EP) \cite{EP}. Inspired by \cite{BP,Unify_1,Unify_2,HVMP}, we employ the Bethe method for a variational approximation of the posterior distribution $p( \bsm{\xi}_q,\{\bsm{\xi}_t^I\} | \{\mbf y_t\})$. Denote by $b(\bsm{\xi}_q,\{\bsm{\xi}_t^I\})$ the surrogate pdf. We aim to find $b(\bsm{\xi}_q,\{\bsm{\xi}_t^I\})$ that minimizes the \emph{variational free energy} between $p( \bsm{\xi}_q,\{\bsm{\xi}_t^I\} | \{\mbf y_t\})$ and $b(\bsm{\xi}_q,\{\bsm{\xi}_t^I\})$:
\begin{align}
	{F}(b) \triangleq \mrm{KL} \big[ & b(\bsm{\xi}_q,\{\bsm{\xi}_t^I\}) || p( \bsm{\xi}_q,\{\bsm{\xi}_t^I\} | \{\mbf y_t\}) \big]  - \mrm{ln} Z, 
	\label{F_b}
\end{align}
where $\mrm{KL}[ \cdot || \cdot ]$ denotes the Kullback–Leibler divergence, and $ - \ln Z = - \ln p(\{\mbf y_t \})$ is termed as the \emph{Helmholtz free energy}. We next give a factorization of $b(\bsm{\xi}_q,\{\bsm{\xi}_t^I\})$. Specifically, Fig. \ref{FG} represents the probability factorization \eqref{p_post} with a factor graph, where the rectangles (termed as factor nodes) represent the probability distributions and the circles (termed as variable nodes) represent the random variables. For conciseness, we introduce the abbreviations of the pdfs, e.g., $p(\mbf y_t |\cdot)$ is abbreviated as $f_{\mbf y_t}$. For any variable node $v$, we introduce the surrogate pdf $b_{v}(v)$ to approximate the posterior pdf of $v$. For example, $b_{\bar{\tau}_{q,l}}(\bar{\tau}_{q,l})$ at variable node $\bar{\tau}_{q,l}$ is an approximation of $p(\bar{\tau}_{q,l}|\{\mbf y_t\})$. For any factor node $f$ and the linked variables represented by vector $\bsm u_f$, we introduce $b_{f}(\bsm u_f)$ to approximate the posterior pdf of $\bsm u_f$. These two types of surrogate pdfs are termed as \emph{beliefs} \cite{BP}, which corresponds to the probability factorization of $b(\bsm{\xi}_q,\{\bsm{\xi}_t^I\})$ as 
\begin{align}
	b(\bsm{\xi}_q,\{\bsm{\xi}_t^I\}) = {\prod_{f \in \mal{F}} b_f(\bsm u_f) } / \left({\prod_{\bsm{v}\in\mal{V}} [ b_{v}(v)]^{A_i-1}} \right), \label{b_fac}
\end{align}
where $A_i$ represents the number of factor nodes connected with $v$; $\mal{F}$ and $\mal{V}$ denote the sets of factor and variable nodes, respectively. $b_f(\bsm u_f)$ and $b_{v}(v)$ are required to fulfill the marginalization consistency:
\begin{align}
	b_{v}(v) = \int b_f(\bsm u_f) d \bsm u_f \backslash v \label{mar_cons}
\end{align} However, the high-dimensional integral in \eqref{mar_cons} probably leads to intractable computation. To avoid this issue, we relax the marginalization consistency to the matching of numerical characteristics. In specific, for the beliefs $b_f(\bsm u_f),f \in \{ f_{\beta_{t,l}},f_{\beta_{t,k,l}^I} \}$, we introduce the matching of the first- and second-moments as \cite{EP} 
\begin{subequations}
	\label{cons_1}
	\begin{align}
		& \mathbb{E}[v| b_{v}(v)] = \mathbb{E}\left[ v | \int b_f(\bsm u_f) d \bsm u_f \backslash v \right], \\
		& \mathbb{E}[ |v|^2 | b_{v}(v)] = \mathbb{E}\left[ |v|^2
		 | \int b_f(\bsm u_f) d \bsm u_f \backslash v \right].
	\end{align}	
\end{subequations} For $b_f(\bsm u_f),f \in \{ f_{\bar{\tau}_{q,l}},f_{\bar{\theta}_{q,l}},f_{\bar{\phi}_{q,l}},f_{{\tau}_{t,k,l}^I},f_{{\theta}_{t,k,l}^I},f_{{\phi}_{t,k,l}^I} \}$, we introduce the matching of mean direction as
\begin{subequations}
	\label{cons_2}
	\begin{align}
		& \mathbb{E}[\cos(v)| b_{v}(v)] = \mathbb{E}\left[ \cos(v) | \int b_f(\bsm u_f) d \bsm u_f \backslash v \right], \\
		& \mathbb{E}[ \sin(v) | b_{v}(v)] = \mathbb{E}\left[ \sin(v)
		 | \int b_f(\bsm u_f) d \bsm u_f \backslash v \right].
	\end{align}	
\end{subequations} For the belief $b_{f_{\mbf y_t}}(\bsm u_f)$, it associates with $f_{\mbf y_t}$ as a multi-linear model of $\beta_{t,l}$, $\beta_{t,k,l}^I$, $\mbf a_N(\cdot)$, $\mbf a_{M_1}(\cdot)$, and $\mbf a_{M_2}(\cdot)$. This makes the computation of matching condition involve prohibitively high complexity. Following \cite[Chap. 10.1.1]{bishop2006pattern}, we introduce a factorization of $b_{f_{\mbf y_t}}(\bsm u_f)$ as
\begin{align}
	b_{f_{\mbf y_t}}(\bsm u_f) = b_{f_{\mbf y_t}}(\bsm{\beta}_t^{uI}) \prod_{v \in \mal{N}_{\mbf y_t}  \backslash \bsm{\beta}_t^{uI}} b_{f_{\mbf y_t}}( v), \label{cons_3} 
\end{align} where $\bsm{\beta}_t^{uI} = [\bsm{\beta}_{t,1},...,\bsm{\beta}_{t,\bar{L}},(\bsm{\beta}_{t,1}^I)^T,...,(\bsm{\beta}_{t,K}^I)^T]^T$; $\mal{N}_{\mbf y_t} = \{\bsm{\beta}_t^{uI},\bar{\tau}_{q,l},\bar{\theta}_{q,l},\bar{\phi}_{q,l},{\tau}_{t,k,l}^I,{\theta}_{t,k,l}^I,{\phi}_{t,k,l}^I \}$. With the factorization \eqref{cons_3}, the original marginalization consistency for $b_{f_{\mbf y_t}}(\bsm u_f)$ reduces to that for $b_{f_{\mbf y_t}}(\bsm{\beta}_t^{uI})$ and $b_{f_{\mbf y_t}}(v),v \in \mal{N}_{\mbf y_t} \backslash \bsm{\beta}_t^{uI}$. Although $b_{f_{\mbf y_t}}(\bsm{\beta}_t^{uI})$ is a joint distribution, we demonstrate below that $b_{f_{\mbf y_t}}(\bsm{\beta}_t^{uI})$ follows a Gaussian distribution. Then, satisfying the marginal consistency for $b_{f_{\mbf y_t}}(\bsm{\beta}_t^{uI})$ reduces to matching the mean and the covariance.

% with the marginalization consistency $b_{v}(v) = b_{f_{\mbf y_t},v}(v)$. 
Given \eqref{b_fac}, \eqref{cons_1}-\eqref{cons_3}, the stationary points of the minimization problem of \eqref{F_b} satisfy the fixed-point equations as 
\begin{subequations}
	\label{belief}
	\begin{align}
		& b_{v}(v) \propto \prod_{f \in \mal{N}(v)} \mal{M}_{ {f} \to v}(v),
		\label{b_v} \\
		& b_{f}(\bsm u_f) \propto f(\bsm u_f)  \prod_{v \in \mal{N}(f)} \mal{M}_{v \to f}(v), f \in \mal{F} \backslash f_{\mbf y_t}, \label{b_f} \\
		& b_{f_{\mbf y_t}}(v) \propto e^{\int b_{f_{\mbf y_t}}(\bsm{\beta}_t^{uI}) \prod_{\tilde{v} \in \mal{N}_{f_{\mbf y_t}} \backslash \{\bsm{\beta}_t^{uI},v \}} b_{f_{\mbf y_t}}(\tilde{v}) \ln {f_{\mbf y_t}}(\bsm u_f) d \bsm u_f \backslash v} \notag \\ 
		& \hspace{40pt} \times \mal{M}_{v \to f_{\mbf y_t}}(v), v \in \mal{N}_{\mbf y_t} \backslash \bsm{\beta}_t^{uI}, \label{b_y_v}  \\
		& b_{f_{\mbf y_t}}(\bsm{\beta}_t^{uI}) \propto e^{\int \prod_{\tilde{v} \in \mal{N}(f_{\mbf y_t}) \backslash \bsm{\beta}_t^{uI}} b_{f_{\mbf y_t}}(\tilde{v}) \ln {f_{\mbf y_t}}(\bsm u_f) d \bsm u_f \backslash \bsm{\beta}_t^{uI}} \notag \\
		& \hspace{52pt} \times \prod_{v \in \bsm{\beta}_t^{uI}} \mal{M}_{v \to f_{\mbf y_t}}(v), \label{b_y_beta}
	\end{align}
\end{subequations}
with
\begin{subequations}
\label{mg_rule}
\begin{align}
	& \mal{M}_{v \to f}(v) \propto \prod_{\tilde{f} \in \mal{N}(v) \backslash f } \mal{M}_{ \tilde{f} \to v} (v), \label{mp_a} \\
	&  \mal{M}_{f \to v} (v) \propto {\mal{P} \Big(  \int f(\bsm u_f) \prod_{\tilde{v} \in \mal{N}(f)} \mal{M}_{\tilde{v} \to f}(\tilde{v})  d \bsm u_f \backslash v \Big)} \notag \\
	& \hspace{2cm} / \mal{M}_{v \to f}(v), f \in \mal{F} \backslash f_{\mbf y_t}, \label{mp_b} \\
	& \mal{M}_{f_{\mbf y_t} \to v} (v) \propto b_{f_{\mbf y_t}}(v) / \mal{M}_{v \to f_{\mbf y_t}}(v), v \in \mal{N}_{\mbf y_t} \backslash \bsm{\beta}_t^{uI}, \label{mp_c} \\
	& \mal{M}_{f_{\mbf y_t} \to v} (v) \propto \left( \int b_{f_{\mbf y_t}}(\bsm{\beta}_t^{uI}) d \bsm{\beta}_t^{uI} / v \right) / \mal{M}_{v \to f_{\mbf y_t}}(v) , \notag \\
	& \hspace{6.2cm} v \in \bsm{\beta}_t^{uI}, \label{mp_d} 
\end{align}	
\end{subequations}
where $\mal{P}(b(\cdot))$ denotes the projection of $b(\cdot)$ to a Gaussian or V-M pdf satisfying \eqref{cons_1} or \eqref{cons_2}; $\mal{N}(f)$ is the set of variables associated with the factor node $f$; $\mal{N}(f) \backslash v$ denotes the set $\mal{N}(f)$ excluding the element $v$; $\mal{N}(v)$ is the set of factor nodes associated with the variable node $v$; $\mal{N}(v) \backslash f$ denotes the set $\mal{N}(v)$ excluding the factor node $f$; $\mal{M}_{v \to f}(v)$ (or $\mal{M}_{f \to v} (v)$) is termed as the \emph{message} from the variable node $v$ to the factor node $f$ (or from the factor node $f$ to the variable node $v$). 

In \eqref{mp_a}, $\mal{M}_{v \to f_{\mbf y_t}}(v),v \in \bsm{\beta}_t^{uI}$ is Gaussian due to the Gaussian projector $\mal{P}(\cdot)$ in \eqref{mp_b}. In \eqref{b_y_beta}, $\mrm{exp}\Big(\int \prod_{\tilde{v} \in \mal{N}(f_{\mbf y_t}) \backslash \bsm{\beta}_t^{uI}} b_{f_{\mbf y_t}}(\tilde{v}) \ln {f_{\mbf y_t}}(\bsm u_f) d \bsm u_f \backslash \bsm{\beta}_t^{uI} \Big)$ is also Gaussian, since $\ln {f_{\mbf y_t}}(\bsm u_f)$ is quadratic w.r.t. $\bsm{\beta}_t^{uI}$ from \eqref{y_CKM}. Then, based on the fact that the product of Gaussian pdfs is still a Gaussian pdf (up to a scaling constant), $b_{f_{\mbf y_t}}(\bsm{\beta}_t^{uI})$ in \eqref{b_y_beta} is Gaussian.

The proof of \eqref{b_v}-\eqref{mp_c} mainly follows \cite[Appendix A]{HVMP} and is omitted due to limited space. In \eqref{b_v}-\eqref{mp_c}, the messages related to $\beta_{t,l}$ are projected to Gaussian pdf, which is widely used in the existing message-passing algorithm \cite{EP,Unify_2,HVMP}. The intuition is that for the power-constraint variable $\beta_{t,l}$, the Gaussian distribution is the maximum entropy distribution. In other words, for a random variable with a power constraint, Gaussian pdf is the best choice in the principle of maximum entropy \cite{guiasu1985principle}. Similarly, we introduce the projection to V-M pdf for the messages related to circular variables $\{\bar{\tau}_{q,l},\bar{\theta}_{q,l},\bar{\phi}_{q,l},{\tau}_{t,k,l}^I,{\theta}_{t,k,l}^I,{\phi}_{t,k,l}^I \}$, since the V-M pdf is the maximum entropy distribution for the circular variable given mean direction and mean resultant length \cite[Chap. 3]{mardia2009directional}.  

% Compared to the message-passing in \cite{HVMP}, a modification is made to better solve the problem of CKM construction. Specifically, we introduce the matching condition of mean-direction \eqref{cons_2} for circular variables $\{\bar{\tau}_{q,l},\bar{\theta}_{q,l},\bar{\phi}_{q,l},{\tau}_{t,k,l}^I,{\theta}_{t,k,l}^I,{\phi}_{t,k,l}^I \}$, which make the corresponding messages follow V-M pdf. The V-M pdf is the maximum entropy distribution (with maximum information uncertainty) for circular variable with given mean direction and mean resultant length \cite[Chap. 3]{mardia2009directional}. The projection to a V-M distribution in (19) enhances the robustness of message passing.

% Second, we maintain the marginalization consistency for the factor node $f_{\mbf y_t}$.  

\subsection{Detailed Derivation of Messages}

From \eqref{b_v} and \eqref{b_f}, $b_{v}(v)$ and $b_{f}(\bsm u_f)$ are determined given $\mal{M}_{ {f} \to v}(v)$ and $\mal{M}_{v \to f}(v)$, respectively. From \eqref{mp_a}, $\mal{M}_{v \to f}(v)$ is determined given $\mal{M}_{ {f} \to v}(v)$. Then, it suffices to update $\mal{M}_{ {f} \to v} (v)$ and $b_{f_{\mbf y_t}}(v)$ in each message-passing iteration. We next present the detailed operations following the message-passing rule \eqref{belief}-\eqref{mg_rule}.

\subsubsection{Messages related to $\bar{\tau}_{q,l}$, $\bar{\theta}_{q,l}$, $\bar{\phi}_{q,l}$, ${\tau}_{t,k,l}^I$, ${\theta}_{t,k,l}^I$, and ${\theta}_{t,k,l}^I$}

According to \eqref{mp_c}, we update $\mal{M}_{f_{\mbf y_t} \to \bar{\tau}_{q,l}} (\bar{\tau}_{q,l})$ as
\begin{align}
	& \mal{M}_{f_{\mbf y_t} \to \bar{\tau}_{q,l}} (\bar{\tau}_{q,l}) \notag \\ 
	& \propto \mrm{exp}\left( \int \prod_{v \in \mal{N}(f_{{\mbf y}_t}) \backslash \bar{\tau}_{q,l}} 
	b_{f_{\mbf y_t}}(v) \ln {f_{\mbf y_t}}(\bsm u_{f_{\mbf y_t}}) d \bsm u_{f_{\mbf y_t}} \backslash \bar{\tau}_{q,l} \right). \label{M_y_to_tau}
\end{align}
In \eqref{M_y_to_tau}, $b_{f_{\mbf y_t}} (v),v \in \{\bar{\tau}_{q,l}, \bar{\theta}_{q,l},\bar{\phi}_{q,l},\tau_{t,k,l}^I,\theta_{t,k,l}^I,\phi_{t,k,l}^I \} $ are shown to be V-M pdfs with mean directions ${\mu}_v$ and concentrations $\kappa_v$ in the remainder of this subsection.
% To compute $\mal{M}_{f_{\mbf y_t} \to \bar{\tau}_{q,l}} (\bar{\tau}_{q,l})$ in \eqref{M_y_to_tau}, $b_{f_{\mbf y_t}}(v)$ is required for $v \in \{\bar{\tau}_{q,l}, \bar{\theta}_{q,l},\bar{\phi}_{q,l},\tau_{t,k,l}^I,\theta_{t,k,l}^I,\phi_{t,k,l}^I,\bsm{\beta}_{t}^{uI} \}$. We show later that $b_{f_{\mbf y_t}} (v)$ for $v \in \{\bar{\tau}_{q,l}, \bar{\theta}_{q,l},\bar{\phi}_{q,l},\tau_{t,k,l}^I,\theta_{t,k,l}^I,\phi_{t,k,l}^I \} $ are approximated by V-M pdfs with mean directions ${\mu}_v$ and concentrations $\kappa_v$. 
Besides, we show that $b_{f_{\mbf y_t}} (\bsm{\beta}_t^{uI})$ is Gaussian with mean $\bsm{\mu}_{{\bsm{\beta}}_t^{uI}} = [\mu_{\beta_{t,1}},...,\mu_{\beta_{t,L}},\bsm{\mu}_{\bsm{\beta}_{t,1}^I}^T,...,\bsm{\mu}_{\bsm{\beta}_{t,K}^I}^T]^T$ in \eqref{V_b_fy_beta}. By substituting $b_{f_{\mbf y_t}}(v),v \in \mal{N}(f_{{\mbf y}_t}) \backslash \bar{\tau}_{q,l}$ and $b_{f_{\mbf y_t}} (\bsm{\beta}_t^{uI})$ into \eqref{M_y_to_tau} and conducting integral operations, we obtain 
\begin{align}
	& \mal{M}_{f_{\mbf y_t} \to \bar{\tau}_{q,l}} (\bar{\tau}_{q,l}) \propto \mrm{exp}\Bigg( \frac{2}{\gamma} \Big( \operatorname{Re} \Big \{  \big( \mbf R_t^T \hat{\mbf b}_l^* 
	\! - \! \sum_{j \neq l} \mu_{\beta_{t,j}} \hat{\mbf b}_l^H \hat{\mbf b}_j \hat{\mbf s}_{t,j}\notag \\
	& \hspace{1.1cm} - \! \sum_{k,j} \mu_{\beta_{t,k,j}^I} \hat{\mbf b}_l^H \hat{\mbf b}_{t,k,j}^I \hat{\mbf s}_{t,k,j}^I   \big)^H 
	\mu_{\beta_{t,l}} \mbf s_t(\bar{\tau}_{q,l}) \Big \} \Big) \! \Bigg), \label{M_y_to_tau2}
\end{align}
where $\mbf R_t \in \mbb{C}^{M \times N}$ is the matrixization of $\mbf y_t \in \mbb{C}^{MN}$, 
$\hat{\mbf b}_l =  \mbb E_{b_{f_{\mbf y_t}}(\bar{\theta}_{q,l})} [ \mbf a_{M_1}(\bar{\theta}_{q,l})] \otimes \mbb E_{b_{f_{\mbf y_t}}(\bar{\phi}_{q,l})} [\mbf a_{M_2}(\bar{\phi}_{q,l})] $, 
$\hat{\mbf b}_{t,k,l}^I =  \mbb E_{b_{f_{\mbf y_t}}(\theta_{t,k,l}^I)} [ \mbf a_{M_1}(\theta_{t,k,l}^I)] \otimes \mbb E_{b_{f_{\mbf y_t}}(\phi_{t,k,l}^I)} [\mbf a_{M_2}(\phi_{t,k,l}^I)] $, 
$\hat{\mbf s}_{t,l} = \mbf x_{t,l} \odot \mbb E_{b_{f_{\mbf y_t}(\bar{\tau}_{q,l})}} [\mbf a_{N}(\bar{\tau}_{q,l}) ]$, 
and $\hat{\mbf s}_{t,k,l}^I = \mbf x_{t,k}^I \odot \mbb E_{b_{f_{\mbf y_t}(\tau_{t,k,l}^I)}} [\mbf a_{N}(\tau_{t,l,k}^I) ]$. 
Note that for a V-M variable $\omega \sim \mal{VM}(\mu,\kappa)$, we have $\mbb E_{\omega}[\mbf a_x(\omega)] = (I_0(\kappa))^{-1}[e^{i0\mu}I_0(\kappa),...,e^{i(x-1)\mu}I_{x-1}(\kappa) ]^T$ \cite{mardia2009directional}, where $I_x(\cdot)$ is the modified Bessel function of the first kind and order $x$.
With $\mal{M}_{f_{\mbf y_t} \to \bar{\tau}_{q,l}} (\bar{\tau}_{q,l})$, we apply \eqref{mp_a} to obtain  $\mal{M}_{\bar{\tau}_{q,l} \to f_{\bar{\tau}_{q,l}}}(\bar{\tau}_{q,l}) \propto \prod_{t} \mal{M}_{f_{\mbf y_t} \to \bar{\tau}_{q,l}} (\bar{\tau}_{q,l})$. Then, by using \eqref{mp_b}, the message from $f_{\bar{\tau}_{q,l}}$ to $\bar{\tau}_{q,l}$ is expressed as 
\begin{align}
	\mal{M}_{f_{\bar{\tau}_{q,l}} \to \bar{\tau}_{q,l}}(\bar{\tau}_{q,l}) \propto \frac{\mal{P} \Big( p(\bar{\tau}_{q,l}) \prod_{t} \mal{M}_{f_{\mbf y_t} \to \bar{\tau}_{q,l}}(\bar{\tau}_{q,l}) \Big)}{\prod_{t} \mal{M}_{f_{\mbf y_t} \to \bar{\tau}_{q,l}} (\bar{\tau}_{q,l})}. \label{M_ftau_to_tau}
\end{align}
Note that the complicated form of $\prod_{t} \mal{M}_{f_{\mbf y_t} \to \bar{\tau}_{q,l}} (\bar{\tau}_{q,l})$ makes $\mal P(\cdot)$ in \eqref{M_ftau_to_tau} have no analytical solution. For simplification, we approximate $\prod_t \mal{M}_{f_{\mbf y_t} \to \bar{\tau}_{q,l}}(\bar{\tau}_{q,l})$ as a V-M pdf. Specifically, we adopt the method in \cite[Sec. IV-D]{VALSE} to search a local maximum of $\prod_t \mal{M}_{f_{\mbf y_t} \to \bar{\tau}_{q,l}}(\bar{\tau}_{q,l})$, denoted by $\hat{\tau}_{q,l}$. We adopt a second-order Taylor expansion of $f(\bar{\tau}_{q,l}) \triangleq \ln(\prod_t \mal{M}_{f_{\mbf y_t} \to \bar{\tau}_{q,l}})$ at $\bar{\tau}_{q,l} = \hat{\tau}_{q,l}$, which is quadratic w.r.t. $\bar{\tau}_{q,l}$, resulting in a Gaussian approximation of $\prod_t \mal{M}_{f_{\mbf y_t} \to \bar{\tau}_{q,l}}$. Then, by using the similarity between the Gaussian pdf and the V-M pdf \cite[Sec. 3.5]{mardia2009directional}, we obtain
\begin{align}
	\prod_t \mal{M}_{f_{\mbf y_t} \to \bar{\tau}_{q,l}}(\bar{\tau}_{q,l}) \approx \mal{VM}(\bar{\tau}_{q,l} ; \mu_{\bar{\tau}_{q,l} \to f_{\bar{\tau}_{q,l}}},\kappa_{\bar{\tau}_{q,l} \to f_{\bar{\tau}_{q,l}}}), \label{VM_fy_to_tau}
\end{align}
with
\begin{subequations}
	\label{M_ka_fy_to_tau}
	\begin{align}
		& \mu_{\bar{\tau}_{q,l} \to f_{\bar{\tau}_{q,l}}} = \hat{\tau}_{q,l} - f'(\hat{\tau}_{q,l})/f''(\hat{\tau}_{q,l}), \\
		& \kappa_{\bar{\tau}_{q,l} \to f_{\bar{\tau}_{q,l}}} = \mal{A}^{-1}(e^{-0.5 f''(\hat{\tau}_{q,l})}),
	\end{align}
\end{subequations}
where $\mal{A}^{-1}(\cdot)$ represents the inverse of the function $\mal{A}(\cdot) = I_1(\cdot)/I_0(\cdot)$. By substituting \eqref{VM_fy_to_tau} into \eqref{M_ftau_to_tau} and noting that the product of two V-M pdfs is still a V-M pdf, we obtain 
\begin{align}
	\mal{M}_{f_{\bar{\tau}_{q,l}} \to \bar{\tau}_{q,l}}(\bar{\tau}_{q,l}) =  p(\bar{\tau}_{q,l}) = \mal{VM}(\bar{\tau}_{q,l};\mu_{\bar{\tau}_{q,l}}^{\mrm{pri}},\kappa_{\bar{\tau}_{q,l}}^{\mrm{pri}} ). \label{M_ftau_to_tau2}
\end{align}
With \eqref{VM_fy_to_tau} and \eqref{M_ftau_to_tau2}, we express the belief $b_{f_{\mbf y_t}}(\bar{\tau}_{q,l})$ as 
% the belief at variable node $\bar{\tau}_{q,l}$ is expressed as
\begin{subequations}
	\label{26}
	\begin{align}
		b_{f_{\mbf y_t}}(\bar{\tau}_{q,l}) & \propto \mal{M}_{f_{\mbf y_t} \to \bar{\tau}_{q,l}} \mal{M}_{\bar{\tau}_{q,l} \to f_{\mbf y_t}}, \label{b_tau1}\\
		& \propto \mal{M}_{f_{\mbf y_t} \to \bar{\tau}_{q,l}} \prod_{j \neq t} \mal{M}_{f_{\mbf y_j} \to \bar{\tau}_{q,l}} \mal{M}_{f_{\bar{\tau}_{q,l}} \to \bar{\tau}_{q,l}}, \label{b_tau2}\\
		& \propto \mal{VM}(\bar{\tau}_{q,l} ; \mu_{\bar{\tau}_{q,l} \to f_{\bar{\tau}_{q,l}}},\kappa_{\bar{\tau}_{q,l} \to f_{\bar{\tau}_{q,l}}} ) p(\bar{\tau}_{q,l}), \label{b_tau3} \\
		& \propto \mal{VM}(\bar{\tau}_{q,l};\mu_{\bar{\tau}_{q,l}},\kappa_{\bar{\tau}_{q,l}} ), \label{b_tau}
	\end{align}
\end{subequations}
with $\kappa_{\bar{\tau}_{q,l}} =  | \mu_{\bar{\tau}_{q,l} \to f_{\bar{\tau}_{q,l}}} e^{-j\kappa_{\bar{\tau}_{q,l} \to f_{\bar{\tau}_{q,l}}}} + {\mu}_{\bar{\tau}_{q,l}}^{\mrm{pri}} e^{-j{\kappa}_{\bar{\tau}_{q,l}}^{\mrm{pri}}} | $ and $\mu_{\bar{\tau}_{q,l}} = \mrm{arg}(\mu_{\bar{\tau}_{q,l} \to f_{\bar{\tau}_{q,l}}} e^{-j\kappa_{\bar{\tau}_{q,l} \to f_{\bar{\tau}_{q,l}}}} + {\mu}_{\bar{\tau}_{q,l}}^{\mrm{pri}} e^{-j{\kappa}_{\bar{\tau}_{q,l}}^{\mrm{pri}}})$. \eqref{b_tau1} follows from \eqref{mp_c}, and \eqref{b_tau2} follows from \eqref{mp_a}.
% and \eqref{b_tau3} follows from \eqref{VM_fy_to_tau} and \eqref{M_ftau_to_tau2}.

We next update the messages related to $\bar{\theta}_{q,l}$. Following the steps in \eqref{M_y_to_tau}-\eqref{M_y_to_tau2}, we obtain 
\begin{align}
	& \mal{M}_{f_{\mbf y_t} \to \bar{\theta}_{q,l}} (\bar{\theta}_{q,l}) \propto \mrm{exp} \Bigg( \frac{2}{\gamma} \Big( \operatorname{Re} \Big \{  \big( \mbf R_t \hat{\mbf s}_{t,l}^* - \sum_{j \neq l} \mu_{\beta_{t,j}} \hat{\mbf s}_{t,l}^H \hat{\mbf s}_{t,j} \hat{\mbf b}_{j}\notag \\
	& \hspace{0.7cm} - \sum_{k,j} \mu_{\beta_{t,k,j}^I} \hat{\mbf s}_{t,l}^H \hat{\mbf s}_{t,k,j}^I \hat{\mbf b}_{t,k,j}^I \big)^H  \mu_{\beta_{t,l}} \mbf a_{M_1}(\bar{\theta}_{q,l}) \otimes \hat{\mbf a}_{M_2} \Big \} \Big) \Bigg).	\label{M_fy_to_theta}
\end{align}
Similar to \eqref{M_ftau_to_tau}-\eqref{26}, we obtain the belief $b_{f_{\mbf y_t}}(\bar{\theta}_{q,l}) = \mal{VM}(\bar{\theta}_{q,l},\mu_{\bar{\theta}_{q,l}},\kappa_{\bar{\theta}_{q,l}} )$. Note that $\bar{\theta}_{q,l}$ and $\bar{\phi}_{q,l}$ are symmetrical in the signal model \eqref{y_CKM}. Then, $b_{\bar{\phi}_{q,l}}(\bar{\phi}_{q,l})$ is obtained by following the steps for $b_{f_{\mbf y_t}}(\bar{\theta}_{q,l})$. 
% $\mal{M}_{f_{\mbf y_t} \to \bar{\phi}_{q,l}}$ is obtained by replacing the expectation over $\bar{\theta}_{q,l}$ with the expectation over $\bar{\phi}_{q,l}$ in \eqref{M_fy_to_theta}. Then, similar to \eqref{M_ftau_to_tau}-\eqref{b_tau}, we obtain the belief $b_{\bar{\phi}_{q,l}}(\bar{\phi}_{q,l}) = \mal{VM}(\bar{\phi}_{q,l},\hat{\mu}_{\bar{\phi}_{q,l}},\hat{\kappa}_{\bar{\phi}_{q,l}} )$.
As for the message computation of channel parameters of interferences, we note that 
 ${\tau}_{t,k,l}$, ${\theta}_{t,k,l}$, and ${\phi}_{t,k,l}$ are symmetrical to $\bar{\tau}_{q,l}$, $\bar{\theta}_{q,l}$, and $\bar{\phi}_{q,l}$, respectively, when $|\mbb T_q|=1$. The corresponding messages are obtained from a simplified version of updating the messages of $\bar{\tau}_{q,l}$, or more specifically, by removing the message product at different time-slots in \eqref{M_ftau_to_tau}-\eqref{VM_fy_to_tau}.

\subsubsection{Messages related to $\bsm{\beta}_t$ and $\bsm{\beta}_{t,k}^I$}

For inferring $\bsm{\beta}_t$ and $\bsm{\beta}_{t,k}^I$, we express \eqref{y_CKM} as a linear model of $\bsm{\beta}_t^{uI}$:
\begin{align}
	\mbf y_t =  \bsm{\Psi}_t \bsm{\beta}_t^{uI}  + \mbf n_t, \label{y_for_beta}
\end{align}
where $\bsm{\Psi}_t = [\mbf s_t(\bar{\tau}_{q,1}) \otimes \mbf b(\bar{\theta}_{q,1},\bar{\phi}_{q,1}),...,\mbf s_t(\bar{\tau}_{q,\bar{L}}) \otimes \mbf b(\bar{\theta}_{q,\bar{L}},\bar{\phi}_{q,\bar{L}}),\mbf s_t^I(\tau_{t,1,1}^I) \otimes \mbf b(\theta_{t,1,1}^I,\phi_{t,1,1}^I),...,\mbf s_t^I(\tau_{t,K,\bar{L}_K^I}^I)$ $\otimes \mbf b(\theta_{t,K,\bar{L}_K^I}^I,\phi_{t,K,\bar{L}_K^I}^I) ] $. Given \eqref{y_for_beta}, we use \eqref{b_y_beta} to obtain 
% $b_{f_{\mbf y_t}}(\bsm \beta_t^{uI})$ as
\begin{subequations}
	\begin{align}
		& b_{f_{\mbf y_t}}(\bsm \beta_t^{uI}) \notag \\
		& \propto e^{\int \prod_{v \in \mal{N}(f_{\mbf y_t}) \backslash \bsm{\beta}_t^{uI}} b_{f_{\mbf y_t}}(v) \ln {f_{\mbf y_t}}(\bsm u_f) d \bsm u_f \backslash \bsm{\beta}_t^{uI}} \prod_{v \in \bsm{\beta}_t^{uI}} \mal{M}_{f_{v} \to v }(v)  \notag \\
		& \propto \mal{CN} \left( \mathbb{E}\left[ \bsm{\Psi}_t^H \bsm{\Psi}_t \right]^{-1} \hat{\bsm{\Psi}}_t^H \mbf y_t , \gamma \mathbb{E}\left[ \bsm{\Psi}_t^H \bsm{\Psi}_t \right]^{-1} \right) \notag \\
		& \hspace{0.5cm} \times \mal{CN}(\bsm{\mu}_{f_{\bsm{\beta}_t^{uI}} \to \bsm{\beta}_{t}^{uI}},\bsm{\Sigma}_{ f_{\bsm{\beta}_t^{uI}}  \to\bsm{\beta}_{t}^{uI} }), \label{b_fy_beta0} \\
		& \propto \mal{CN} \left( \bsm{\mu}_{\bsm{\beta}_{t}^{uI}}, \bsm{\Sigma}_{\bsm{\beta}_{t}^{uI}}\right), \label{b_fy_beta}
	\end{align}
\end{subequations} 
where $\mathbb{E}[\cdot]$ is taken w.r.t. $v \in \{ \bar{\tau}_{q,l}, \bar{\theta}_{q,l},\bar{\phi}_{q,l}, \tau_{t,k,l}^I,\theta_{t,k,l}^I,\phi_{t,k,l}^I \}$ based on $b_{f_{\mbf y_t}}(v)$, and 
\begin{align}
	& \bsm{\mu}_{\bsm{\beta}_{t}^{uI}} = \bsm{\Sigma}_{\bsm{\beta}_{t}^{uI}}\left(\bsm{\Sigma}_{ f_{\bsm{\beta}_t^{uI}}  \to\bsm{\beta}_{t}^{uI} }^{-1} \bsm{\mu}_{f_{\bsm{\beta}_t^{uI}} \to \bsm{\beta}_{t}^{uI}} + \gamma^{-1} \hat{\bsm{\Psi}}_t^H \mbf y_t  \right), \notag \\
	& \bsm{\Sigma}_{\bsm{\beta}_{t}^{uI}} = \left( \bsm{\Sigma}_{ f_{\bsm{\beta}_t^{uI}}  \to\bsm{\beta}_{t}^{uI} }^{-1} + \gamma^{-1} \mathbb{E}\left[ \bsm{\Psi}_t^H \bsm{\Psi}_t \right] \right)^{-1}. \label{V_b_fy_beta}
\end{align}
To reduce the computational complexity, we ignore the non-diagonal elements of $\bsm{\Sigma}_{\bsm{\beta}_{t}^{uI}}$ in each message-passing iteration.

Denote by $\mu_{\beta_{t,l}}$ and $v_{\beta_{t,l}}$ the $l$-th element of $\bsm{\mu}_{\bsm{\beta}_{t}^{uI}}$ and the $l$-th diagonal element of $\bsm{\Sigma}_{\bsm{\beta}_{t}^{uI}}$, respectively. We apply \eqref{mp_d} to obtain the message from $f_{\mbf y_t}$ to $\beta_{t,l}$ as
\begin{align}
	\mal{M}_{f_{\mbf y_t} \to \beta_{t,l}}(\beta_{t,l}) & \propto \frac{\int b_{f_{\mbf y_t}}(\bsm{\beta}_{t}^{uI}) d \bsm{\beta}_{t}^{uI} \backslash \beta_{t,l} }{ \mal{M}_{f_{\beta_{t,l}} \to \beta_{t,l}}(\beta_{t,l}) }\notag \\
	& \propto \mal{CN}( \mu_{f_{\mbf y_t} \to \beta_{t,l}}, v_{f_{\mbf y_t} \to \beta_{t,l}} ), \label{M_fy_to_beta}
\end{align}
where 
\begin{subequations}
	\begin{align}
		& \mu_{f_{\mbf y_t} \to \beta_{t,l}} = v_{f_{\mbf y_t} \to \beta_{t,l}} \left( \frac{\mu_{\beta_{t,l}}} {v_{\beta_{t,l}}} - \frac{\mu_{f_{\beta_{t,l}} \to \beta_{t,l}}}{v_{f_{\beta_{t,l}} \to \beta_{t,l}}} \right), \\
		& v_{f_{\mbf y_t} \to \beta_{t,l}} = \left( v_{\beta_{t,l}}^{-1} - v_{f_{\beta_{t,l}} \to \beta_{t,l}}^{-1} \right)^{-1}.
	\end{align}
\end{subequations}
Given Gaussian $\mal{M}_{f_{\mbf y_t} \to \beta_{t,l}}(\beta_{t,l})$ in \eqref{M_fy_to_beta}, we use \eqref{p_rho} and \eqref{mp_b} to obtain
\begin{align}
	 \mal{M}_{f_{\beta_{t,l}} \to \beta_{t,l}}(\beta_{t,l}) & \propto \frac{\mal{P}\left(\mal{M}_{f_{\mbf y_t} \to \beta_{t,l}}(\beta_{t,l}) p(\beta_{t,l}) \right)}{\mal{M}_{f_{\mbf y_t} \to \beta_{t,l}}(\beta_{t,l})} \notag \\
	& \propto \mal{CN}(0,\bar{\rho}_{q,l}). \label{M_fbeta_to_beta}
\end{align}
% where \eqref{M_fbeta_to_beta} follows from the fact that $\mal{M}_{f_{\mbf y_t} \to \beta_{t,l}}(\beta_{t,l}) p(\beta_{t,l})$ is Gaussian. 

By replacing $\beta_{t,l}$ with $\beta_{t,k,l}^I$ in \eqref{M_fy_to_beta}, we obtain the message from $f_{\mbf y_t}$ to $\beta_{t,k,l}^I$ as
\begin{align}
	\mal{M}_{f_{\mbf y_t} \to \beta_{t,k,l}^I}(\beta_{t,k,l}^I) \propto \mal{CN}( \mu_{f_{\mbf y_t} \to \beta_{t,k,l}^I}, v_{f_{\mbf y_t} \to \beta_{t,k,l}^I} ), \label{M_fy_to_betaI}
\end{align}
Let $\mal{M}_{f_{\mbf y_t} \to \bsm{\beta}_{t,k}^I}(\bsm{\beta}_{t,k}^I) \triangleq \prod_{l=1}^L \mal{M}_{f_{\mbf y_t} \to \beta_{t,k,l}^I}(\beta_{t,k,l}^I) = \mal{CN}( \bsm{\mu}_{f_{\mbf y_t} \to \bsm{\beta}_{t,k}^I}, \bsm{\Sigma}_{f_{\mbf y_t} \to \bsm{\beta}_{t,k}^I} ) $, where the $l$-th entry of $\bsm{\mu}_{f_{\mbf y_t} \to \bsm{\beta}_{t,k}^I}$ is $\mu_{f_{\mbf y_t} \to \beta_{t,k,l}^I}$, and $\bsm{\Sigma}_{f_{\mbf y_t} \to \bsm{\beta}_{t,k}^I}$ is diagonal with the $(l,l)$-th entry being $v_{f_{\mbf y_t} \to \beta_{t,k,l}^I}$. From \eqref{b_f}, we obtain $b_{f_{\bsm{\beta}_{t,k}^I}}(\bsm{\beta}_{t,k}^I)$ as  
\begin{align} 
	& b_{f_{\bsm{\beta}_{t,k}^I}}(\bsm{\beta}_{t,k}^I)  \propto p(\bsm{\beta}_{t,k}^I) \mal{M}_{f_{\mbf y_t} \to \bsm{\beta}_{t,k}^I}(\bsm{\beta}_{t,k}^I)
	\notag \\ 
	& \propto (1- \lambda_{t,k}^{\mrm{post}}) \delta(\bsm{\beta}_{t,k}^I) + \lambda_{t,k}^{\mrm{post}} 
	{\cal CN}( \bsm{\beta}_{t,k}^I ; \bsm{\eta}_{f_{\bsm{\beta}_{t,k}^I}} ; \bsm{\varLambda}_{f_{\bsm{\beta}_{t,k}^I}} ), \label{b_f_betaI}
	\end{align}
	where
	\begin{subequations}
	\label{interB}
	\begin{align}
		 & \lambda_{t,k}^{\mrm{post}}  = \Big(1 + \frac{ (1-\lambda )   {\cal CN}(\mbf 0 ; \bsm{\mu}_{f_{\mbf y_t} \to \bsm{\beta}_{t,k}^I} , \bsm{\Sigma}_{f_{\mbf y_t} \to \bsm{\beta}_{t,k}^I} )} 
		 { \lambda {\cal CN}( \mbf 0 ; \bsm{\mu}_{f_{\mbf y_t} \to \bsm{\beta}_{t,k}^I} ,  \bsm{\Sigma}_{f_{\mbf y_t} \to \bsm{\beta}_{t,k}^I} + \mrm{diag}(\bsm{\rho}_{t,k}^I) ) }\Big)^{-1} \label{31a} \\
		 & \bsm{\varLambda}_{f_{\bsm{\beta}_{t,k}^I}} = (\bsm{\Sigma}_{f_{\mbf y_t} \to \bsm{\beta}_{t,k}^I}^{-1} + \mrm{diag}(\bsm{\rho}_{t,k}^I)^{-1} )^{-1}, \\
		 & \bsm{\eta}_{f_{\bsm{\beta}_{t,k}^I}}  = \bsm{\varLambda}_{f_{\bsm{\beta}_{t,k}^I}} \bsm{\Sigma}_{f_{\mbf y_t} \to \bsm{\beta}_{t,k}^I}^{-1}  \bsm{\mu}_{f_{\mbf y_t} \to \bsm{\beta}_{t,k}^I}.
	\end{align}
	\end{subequations}
	With \eqref{b_f_betaI}, we compute $\mal{P}(\int b_{f_{\bsm{\beta}_{t,k}^I}}(\bsm{\beta}_{t,k}^I) d \bsm{\beta}_{t,k}^I \backslash \beta_{t,k,l}^I) \propto \mal{CN}(\mu_{f_{\beta_{t,k,l}^I}}, v_{f_{\beta_{t,k,l}^I}})$ with 
	\begin{subequations}
		\label{b_betaI}
		\begin{align}
			& \mu_{f_{\beta_{t,k,l}^I}} = \lambda_{t,k}^{\mrm{post}} \eta_{f_{\beta_{t,k,l}^I}}, \\
			& v_{f_{\beta_{t,k,l}^I}} = \lambda_{t,k}^{\mrm{post}} ( |\eta_{f_{\beta_{t,k,l}^I}}|^2 + \varLambda_{f_{\beta_{t,k,l}^I}} ) - |\mu_{f_{\beta_{t,k,l}^I}}|^2.
		\end{align}		
	\end{subequations}
	Given \eqref{M_fy_to_betaI} and \eqref{b_betaI}, $\mal{M}_{{f_{\beta_{t,k,l}^I} \to \beta_{t,k,l}^I}}(\beta_{t,k,l}^I)$ is computed as $\mal{CN}(\mu_{f_{\beta_{t,k,l}^I}}, v_{f_{\beta_{t,k,l}^I}}) / \mal{CN}( \mu_{f_{\mbf y_t} \to \beta_{t,k,l}^I}, v_{f_{\mbf y_t} \to \beta_{t,k,l}^I} ) = \mal{CN}(\mu_{f_{\beta_{t,k,l}^I} \to \beta_{t,k,l}^I},v_{f_{\beta_{t,k,l}^I} \to \beta_{t,k,l}^I})$ with
	\begin{subequations}
		\label{m_f_to_betaI}
		\begin{align}
			& \mu_{f_{\beta_{t,k,l}^I} \to \beta_{t,k,l}^I}  = v_{f_{\beta_{t,k,l}^I} \to \beta_{t,k,l}^I} \left( \frac{\mu_{f_{\beta_{t,k,l}^I}}} {v_{f_{\beta_{t,k,l}^I}}}
			- \frac{\mu_{f_{\mbf y_t} \to \beta_{t,k,l}^I}} {v_{f_{\mbf y_t} \to \beta_{t,k,l}^I}} \right), \\
			& v_{f_{\beta_{t,k,l}^I} \to \beta_{t,k,l}^I}  = \left( v_{f_{\beta_{t,k,l}^I}}^{-1} - v_{f_{\mbf y_t} \to \beta_{t,k,l}^I}^{-1} \right)^{-1}.
		\end{align}
	\end{subequations}
	With \eqref{M_fbeta_to_beta} and \eqref{m_f_to_betaI}, we obtain $\bsm{\Sigma}_{ f_{\bsm{\beta}_t^{uI}}  \to\bsm{\beta}_{t}^{uI} } = \mrm{diag}  ([\bar{\rho}_{q,1},..., \\ \bar{\rho}_{q,\bar{L}},v_{f_{\beta_{t,1,1}^I} \to \beta_{t,1,1}^I},...,v_{f_{\beta_{t,K,\bar{L}_{K}^I}^I} \to \beta_{t,K,\bar{L}_{K}^I}^I}]^T \! )$ and $\bsm{\mu}_{f_{\bsm{\beta}_t^{uI}} \to \bsm{\beta}_{t}^{uI}} \!$ $  =[ 0,...,0,\mu_{f_{\beta_{t,1,1}^I} \to \beta_{t,1,1}^I},...,\mu_{f_{\beta_{t,K,\bar{L}_{K}^I}^I} \to \beta_{t,K,\bar{L}_{K}^I}^I} \! ]^T$ in \eqref{b_fy_beta0}.
	
	% Modify marginal BG to joint BG.

\subsection{EM Learning}

In practice, the prior parameters $\bsm{\omega} =  \{\gamma,\bar{\rho}_{q,l},\rho_{t,k,l}^I,\lambda \}$ are typically unknown. To learn $\bsm{\omega}$, we adopt the expectation maximization (EM) method \cite[Chap. 9]{bishop2006pattern} based on the received signal $\mbf y_t$. The EM learning is formulated as 
\begin{align} 
	 \bsm{\omega}^{(i+1)} = \arg \max_{\bsm{\omega}} \mbb{E} \Big[ & \sum_t p(\mbf y_t | \bsm{\xi}_q,\bsm{\xi}_t^I ; \gamma) + \sum_{t,l} \ln p(\beta_{t,l};\bar{\rho}_{q,l}) \notag \\ 
	& + \sum_{t,k,l} \ln p( \beta_{t,k,l}^I ; \rho_{t,k,l}^I , \lambda ) \Big],  \label{eq_EM}
\end{align}
where $\mathbb{E}[\cdot]$ is taken based on $b_{f_{\mbf y_t}}(\bsm \beta_t^{uI})$ in and $b_{f_{\mbf y_t}}(v),v \in \{ \bar{\tau}_{q,l}, \bar{\theta}_{q,l},\bar{\phi}_{q,l}, \tau_{t,k,l}^I,\theta_{t,k,l}^I,\phi_{t,k,l}^I \}$. Then, we set the derivatives of \eqref{eq_EM} w.r.t. $\gamma$, $\bar{\rho}_{q,l}$, ${\tau}^I$, and $\lambda$ to zeros, which yields
\begin{subequations}
	\begin{align}
		& \hat{\gamma} = \frac{1}{NM|\mbb{T}_q|} \sum_t \Big( \mbf y_t^H \mbf y_t - 2 \Re \{ \mbf y_t^H \hat{\bsm{\Psi}} \bsm{\mu}_{\bsm{\beta}_t^{uI}} \} \notag \\
		& \hspace{2.8cm} + \mrm{tr}( \mbb E[\bsm{\beta}_t^{uI}(\bsm{\beta}_t^{uI})^H] \mbb E[\bsm{\Psi}\bsm{\Psi}^H] ) \Big), \label{EM_a} \\
		& \hat{\rho}_{q,l} = \frac{1}{|\mbb{T}_q|} \sum_t ( |\mu_{\beta_{t,l}}|^2 + v_{\beta_{t,l}} ),\label{EM_b} \\
		& \hat{\rho}_{t,k,l}^I = \frac{ \sum_{t,k} \lambda^{post}_{t,k}  \sum_l(|\mu_{\beta^I_{t,k,l}}|^2+v_{\beta_{t,k,l}^I})/L_k^I } {\sum_{t,k} \lambda^{post}_{t,k}}, \label{EM_c} \\
		& \hat{\lambda} = \frac{1}{|\mbb{T}_q|} \sum_{t,k} {\lambda^{post}_{t,k} }/{|\mal{K}_t|} \label{EM_d}.
	\end{align}
\end{subequations}

\subsection{Overall Algorithm}

We summarize the proposed message-passing algorithm in Algorithm 1. For initialization, one choice is to use $p(v)$ and $b_{f_{\mbf y_t}}(v), v \in \mal{N}_{\mbf y_t}$ from the previous update period of CKM. This strategy may lead to a probability mismatch when the birth/death of channel paths exist. For algorithm robustness, we choose to use a non-informative prior for $p(v), v \in \{ \bar{\tau}_{q,l}, \bar{\theta}_{q,l},\bar{\phi}_{q,l}, \tau_{t,k,l}^I,\theta_{t,k,l}^I,\phi_{t,k,l}^I \}$, i.e., $\mu_v^{\mrm{pri}} \to 0$ and $\kappa_v^{\mrm{pri}} \to 0$ to make $p(v)$ tend to be a uniform pdf in $[0,2\pi)$ \cite{mardia2009directional}. The initial $\mu_v$ for $b_{f_{\mbf y_t}}(v),v \in \{ \bar{\tau}_{q,l}, \bar{\theta}_{q,l},\bar{\phi}_{q,l}, \tau_{t,k,l}^I,\theta_{t,k,l}^I,\phi_{t,k,l}^I \}$
randomly take the values within the region $[0,2\pi)$. By noting that a small concentration leads to $\mbb{E}_v[\mbf a_x(v)] \to \mbf 0$, we set $\kappa_v$ of $b_{f_{\mbf y_t}}(v)$ to a relatively large value, e.g., $10^5$. Then, a least-square initialization of $\bsm{\mu}_{\bsm{\beta}^{uI}_t}$ is adopted as $ \bsm{\mu}_{\bsm{\beta}^{uI}_t} = (\hat{\bsm \Psi}_t^H  \hat{\bsm \Psi}_t)^{-1} \hat{\bsm \Psi}_t^H \mbf y_t, \forall t \in \mbb{T}_q$. Particularly, sparsity $\lambda$ is initialized to $1$, which means that $\bsm{\beta}_{t,k}^I,\forall t,k$ are treated as non-sparse vectors. This setting is because with random initialization of $ \mu_v, v \in \{ \bar{\tau}_{q,l}, \bar{\theta}_{q,l},\bar{\phi}_{q,l}, \tau_{t,k,l}^I,\theta_{t,k,l}^I,\phi_{t,k,l}^I \}$, $\mu_v$ dramatically deviate from the true values in the first iteration, leading to a relatively severe mismatch of $\mathbf a_N(\cdot)$ and $\mbf a_{M_1}(\cdot)\otimes \mbf a_{M_2}(\cdot)$. Due to this mismatch, the magnitude of the initial estimate of $\bsm{\beta}_t^{uI}$ is small, causing the EM learning $\hat{\lambda} \to 0$. To avoid this issue, we set $\lambda = 1$ for the first message-passing. Then, when all parameters have been updated once, we set $\lambda = 0.5$ to exploit the sparse prior in \eqref{p_rho_I}, and the value of $\lambda$ in the following iterations is given by the EM learning. The prior variances of $\beta_{t,l}$ and $\beta_{t,k,l}^I$ are initialized to $\bar{\rho}_{q,l} = 1/\bar{L}$, and $\rho_{t,k,l}^I = 1/ L_k^I$, respectively. Besides, we initialize $\gamma = \sum_{t \in \mbb{T}_q } || \mbf y_t ||_2^2 / (MN|\mbb{T}_q|)$, $\mu_{f_{\beta_{t,k,l}^I} \to \beta_{t,k,l}^I}=0$, $v_{f_{\beta_{t,k,l}^I} \to \beta_{t,k,l}^I}=1/L_k^I$, $\mu_{f_{\beta_{t,k}} \to \beta_{t,k}}=0$, and $v_{f_{\beta_{t,k}} \to \beta_{t,k}}=1/\bar{L}$. 

Algorithm 1 constructs the CKM in the Bayesian inference framework. In particular, step 1 updates the messages of $\bar{\tau}_{q,l}$, $\bar{\theta}_{q,l}$ and $\bar{\phi}_{q,l}$.
% where we use the mean directions of $b_{f_{\mbf y_t}}(v), v \in \{\bar{\tau}_{q,l},\bar{\theta}_{q,l},\bar{\phi}_{q,l}\}$ as the estimated delays and angles. 
Alongside the EM learning of $\bar{\rho}_{q,l}$ in step 4, we construct the CKM mapping $\mal C(\cdot)$ from the algorithm outputs at the final iteration. 
% Compared to existing OMP \cite{OMP} which is typically adopted for parameter extraction, the proposed method outputs continuous parameters rather than discrete ones from a pre-determined dictionary. This is more representative of the continuous nature of actual channel parameters. 
Compared to the compressive-sensing (CS) methods \cite{OMP,AMP_2009,turbo_ma,STCS} designed for an AWGN signal model, the proposed algorithm simultaneously estimates the channel parameters from multiple observations $\{\mbf y_t\}_{t \in \mbb{T}_q}$ and cancel interferences by exploiting the block-sparsity. In Sec. VI, we show that Algorithm 1 can achieve accurate CKM construction at a relatively low SINR, while other methods cannot.

We now analyze the computational complexity of Algorithm 1. In step 1, the dominant computation is the matrix multiplication in \eqref{M_y_to_tau2} and \eqref{M_fy_to_theta} with the complexity $\mal{O}(|\mbb{T}_q|MN\bar{L}+ MN\sum_k \bar{L}_k^I)$. The other operations involve point-wise calculations with a computational complexity linear to the system size. In steps 2-3, the matrix inversions in \eqref{b_fy_beta} and \eqref{V_b_fy_beta} dominate the computational complexity with complexity $\mal{O}(|\mbb{T}_q|(\bar{L}+\sum_k \bar{L}_k^I)^3)$. The EM learning in steps 4-5 involves point-wise operations and matrix-vector multiplications, which are computationally negligible compared to the matrix inversions in steps 2-3. In summary, the computational complexity of Algorithm 1 is $\mal{O}(|\mbb{T}_q|MN\bar{L}+ MN\sum_k \bar{L}_k^I + |\mbb{T}_q|(\bar{L}+\sum_k \bar{L}_k^I)^3 )$ per iteration.

\begin{algorithm}[h]
   \small
   \caption{\label{alg1} Interference-cancellation-based CKM}
   \begin{algorithmic}
	
	   \REQUIRE $\{\mathbf{y}_t\}_{t \in \mbb{T}_q}$, $\{\mbf x_t\}_{t \in \mbb{T}_q}$, and $\{\mbf x_{t,k}^I\}_{t \in \mbb{T}_q,\forall k}$
	   
	   \hspace{-9pt} \textbf{Parameter initialization}  % $\mu_v \leftarrow \mal{U}(0,2\pi), \kappa_v \to \infty, v \in \{ \bar{\tau}_{q,l}, \bar{\theta}_{q,l},\bar{\phi}_{q,l}, \tau_{t,k,l}^I,\theta_{t,k,l}^I,\phi_{t,k,l}^I\}$, $\bsm{\mu}_{\bsm{\beta}^{uI}_t} = (\hat{\bsm \Psi}_t^H  \hat{\bsm \Psi}_t)^{-1} \hat{\bsm \Psi}_t^H \mbf y_t$, $\gamma = \sum_{t \in \mbb{T}_q } || \mbf y_t - \hat{\bsm \Psi}_t \bsm{\beta}^{uI}_t ||_2^2 / |\mbb{T}_q|$, $\bar{\rho}_{q,l} = 1/\bar{L}$, $\rho_{t,k,l}^I = 1/ L$, $\mu_{f_{\beta_{t,k,l}^I} \to \beta_{t,k,l}^I}=0$, $v_{f_{\beta_{t,k,l}^I} \to \beta_{t,k,l}^I}=1/L$, $\mu_{f_{\beta_{t,k}} \to \beta_{t,k}}=0$, $v_{f_{\beta_{t,k}} \to \beta_{t,k}}=1/\bar{L}$, and $\lambda = 1$.
	   % $\beta_{t,l}=0$, $\beta_{t,k,l}^I=0$. $\bar{\tau}_{q,l}$, $\bar{\theta}_{q,l}$, $\bar{\phi}_{q,l}$, ${\tau}_{t,k,l}^I$, ${\theta}_{t,k,l}^I$, and ${\theta}_{t,k,l}^I$ randomly take the values within $[0,2\pi)$.
	  
	   \textbf{For} iteration number $i=1,2,...,I_{\mrm{max}}$, \textbf{do}
	   
	~~ \% Message-passing.

	   \quad \ 1: Update messages of $v \in \{ \bar{\tau}_{q,l}, \bar{\theta}_{q,l},\bar{\phi}_{q,l}, \tau_{t,k,l}^I,\theta_{t,k,l}^I,\phi_{t,k,l}^I \}$
	   
	   \quad \quad ~ by \eqref{M_y_to_tau2}, \eqref{VM_fy_to_tau}-\eqref{M_ka_fy_to_tau}, and \eqref{b_tau}-\eqref{M_fy_to_theta};
	   
	   \quad \ 2: Update messages of $\beta_{t,l}$ by \eqref{V_b_fy_beta} and \eqref{M_fbeta_to_beta}; 

	   \quad \ 3: Update messages of $\beta_{t,k,l}^I$ by \eqref{M_fy_to_betaI} and \eqref{interB}-\eqref{m_f_to_betaI}; 

	   ~~ \% EM learning.

	   \quad \ 4: Update $\gamma$, $\bar{\rho}_{q,l}$, and $\rho_{t,k,l}^I$ by \eqref{EM_a}-\eqref{EM_c};
	   
	   \quad \ 5: Set $\lambda=0.5$ if $i=1$ and update $\lambda$ by \eqref{EM_d} if $i > 1$.

	   \textbf{end}
	   
	   \ENSURE $\bar{\tau}_{q,l} \leftarrow \mu_{\bar{\tau}_{q,l}}$, $\bar{\theta}_{q,l} \leftarrow \mu_{\bar{\theta}_{q,l}}$, $\bar{\phi}_{q,l} \leftarrow \mu_{\bar{\phi}_{q,l}}$, $\bar{\rho}_{q,l} \leftarrow \mu_{\bar{\rho}_{q,l}}$.
	
   \end{algorithmic}
\end{algorithm}

% At the first iteration, we initialize these V-M beliefs with random mean directions and near-zero concentrations, i.e., $1$  

\vspace{-0.5cm}

\section{CKM-Assisted Channel Estimation}

In subsection A, we introduce the minimum mean square error and interference rejection combining (MMSE-IRC) channel estimator \cite{IRC}. In MMSE-IRC, the joint frequency-space covariance of the user channel and the spatial covariance of interference signals are estimated based on the CKM, as shown in subsection B. In subsection C, we reduce the computational complexity of MMSE-IRC by exploiting the structure of user channel covariance. In subsection D, we analyze how the accuracy of user channel covariance impacts the performance of the proposed MMSE-IRC.

\vspace{-0.3cm}
\subsection{MMSE-IRC Estimator}

Consider the channel estimation at current time-slot $t$. Recall that the period of pilot transmission is typically much shorter than the period of CKM update. 
Then, it is challenging for neighbor cells to establish a low-latency sharing of pilot usage information for real-time channel estimation.
% \footnote{In 5g cellular networks, pilots are independently scheduled at different cells without information sharing.} 
For the generality of channel estimator, we treat $\mbf X_{t,k}^I,\forall k$ as unknown interference pilots here.
% But a main problem is that the interference pilot $\mbf X_{t,k}^I$ is typically unknown in channel estimation.\footnote{In 5G, pilots are independently scheduled in different cells without inter-cell cooperation to inform pilot usage information. Achieving low-latency cooperation to assist real-time channel estimation remains an open problem.}
Based on \eqref{Y_t}, we multiply $\mbf Y_t$ by $\mbf X^{-1}$ and vectorize  $\mbf X^{-1}\mbf Y_t$ to obtain
% \begin{align}
%   \mbf X^{-1} \mbf Y_t = \mbf H_t + \sum_{k} \mbf X^{-1} \mbf X_{t,k}^I \mbf H_{t,k}^I + \mbf X^{-1} \mbf N_t, \label{XY}
% \end{align}
%as an additive colored noise model of $\mbf H_t$. Recall that $\mbf X$ is a diagonal pilot matrix with the average power of each diagonal element being $P$. Then, $\mbf X^{-1} \mbf N_t$ is a AWGN matrix with elements independently drawn from $\mal{CN}(0,\gamma /P)$. We next provide a vector-form of \eqref{XY} as
\begin{align}
 \mbf r_t = \mbf h_t + \sum_k \tilde{\mbf h}_{t,k}^I + \mbf w_t, \label{r_t}
\end{align}
where $\mbf r_t = \mrm{vec}(\mbf X^{-1} \mbf Y_t) $; $\mbf h_t = \mrm{vec}(\mbf H_t)$; $\tilde{\mbf h}_{t,k}^I = \mrm{vec}(\mbf X^{-1} \mbf X_{t,k}^I \mbf H_{t,k}^I)$; 
%Recall that $\mbf X$ is a diagonal pilot matrix with the average pilot power being $P$. 
$\mbf w_t$ is a AWGN vector with elements independently drawn from $\mal{CN}(0,\gamma /P)$.  Denote by $\hat{\mbf C}_{t}$ and $\hat{\mbf R}_t^I$ the estimates of the covariances of $\mbf h_t$ and $\sum_k \tilde{\mbf h}_{t,k}^I$, respectively. Then, the MMSE-IRC estimator is expressed as 
\begin{align}
 \hat{\mbf h}_t = \hat{\mbf C}_{t} \left( \hat{\mbf C}_{t} + \hat{\mbf R}_t^I  \right)^{-1} \mbf r_t. \label{IRC}
\end{align}
% Note that the estimator \eqref{IRC} becomes the well-known linear MMSE estimator \cite[Chp. 12]{Kay} when $\hat{\mbf C}_{t}$ and $\hat{\mbf R}_t^I$ equal the covariances of $\mbf h_t$ and $\sum_k \tilde{\mbf h}_{t,k}^I + \mbf w_t$, respectively. 
% To effectively retrieve $\mbf h_t$, we need to accurately estimate the covariance matrix of $\mbf h_t$ and $\tilde{\mbf h}_{t,k}^I$. The latter is, however, typically unknown in practice and needs to be estimated.
The MMSE-IRC performance is affected by $\hat{\mbf C}_{t}$ and $\hat{\mbf R}_t^I$. Particularly, when $\hat{\mbf C}_{t}$ and $\hat{\mbf R}_t^I$ equal the covariances of $\mbf h_t$ and $\sum_k \tilde{\mbf h}_{t,k}^I + \mbf w_t$, \eqref{IRC} becomes the Bayesian optimal linear MMSE estimator \cite[Chp. 12]{Kay}. In the following, we employ the CKM to provide strong prior $\hat{\mbf C}_{t}$ and assist the acquisition of $\hat{\mbf R}_t^I$.  
\vspace{-0.5cm}
\subsection{Covariance Estimates based on CKM}
% We now consider the choices of $\hat{\mbf C}_{t}$ and $\hat{\mbf R}_t^I$. 
To obtain $\hat{\mbf C}_{t}$, we map time-slot $t$ to location $q$ by using $q = \mal{L}(t)$ in \eqref{q_t}. Then, we apply $\mal C(q)$ in \eqref{CKM} to obtain the channel parameters $\{ \bar{\tau}_{q,l},\bar{\theta}_{q,l},\bar{\phi}_{q,l},\bar{\rho}_{q,l} \}_{l=1}^{\bar{L}}$. With some abuse of notation, $\{ \bar{\tau}_{q,l},\bar{\theta}_{q,l},\bar{\phi}_{q,l},\bar{\rho}_{q,l} \}_{l=1}^{\bar{L}}$ are treated as deterministic variables in this section. Then, we apply  \eqref{apr_H_t} to express $\mbf h_t = \mrm{vec}(\mbf H_t)$ as 
\begin{align}
 \mbf h_t \! = \! \sum_l \bar{\alpha}_{t,l} \sqrt{\bar{\rho}_{q,l}} \mbf b(\bar{\theta}_{q,l},\bar{\phi}_{q,l}) \otimes {\mbf a}_{N} (\bar{\tau}_{q,l}) \! + \! \mrm{vec}(\bsm{\Delta}_{\mbf H_t}).
\end{align} 
Recall that $\bar{\alpha}_{t,l},\forall t,l$ are modeled as i.i.d. complex elements with zero means and unit variances. We ignore the representation error $\mrm{vec}(\bsm{\Delta}_{\mbf H_t})$ and compute 
\begin{align}
 \hat{\mbf C}_{t} = \mbb E_{\{\bar{\alpha}_{t,l}\}}[ \mbf h_t \mbf h_t^H ] = \sum_l \bar{\rho}_{q,l} \left(\mbf b \otimes {\mbf a}_{N} \right) \left(\mbf b \otimes {\mbf a}_{N} \right)^H.
\end{align}

To obtain $\hat{\mbf R}_t^I$, we first analyze the structure of ${\mbf R}_t^I$. Specifically, we use \eqref{H_I} and express $\tilde{\mbf h}_{t,k}^I$ as
\begin{align}
 \tilde{\mbf h}_{t,k}^I = \sum_{l=1}^{L_k^I} \alpha_{t,k,l}^I \sqrt{\rho_{t,k,l}^I} \mbf b(\theta_{t,k,l}^I,\phi_{t,k,l}^I) \otimes \mbf s_{t,k,l},
    \label{45}
\end{align}
where $\mbf s_{t,k,l} = \mbf X^{-1} \mbf X_{t,k}^I \mbf a_N(\tau_{t,k,l}^I)$. Since $\mbf X_{t,k}^I$ is unknown, we treat $\mbf X_{t,k}^I$ as a random matrix with $\mbb E[\mbf X_{t,k}^I (\mbf X_{t,k}^I)^H] =  P^I \mbf I$.
% Recall that in channel estimation, $\mbf X_{t,k}^I$ is unknown with $\mbb E[\mbf X_{t,k}^I (\mbf X_{t,k}^I)^H] =  P^I \mbf I$. 
It follows that $\mbb E[ \mbf s_{t,k,l} \mbf s_{t,k,l}^H ] = P^I P^{-1} \mbf I$. 
Then, we obtain
\begin{subequations}
    \begin{align}
 \mbf R_t^I & = \mbb E_{ \{\mbf s_{t,k,l},\alpha_{t,k,l}^I \}} [ (\sum_k \tilde{\mbf h}_{t,k}^I + \mbf w_t)(\sum_k \tilde{\mbf h}_{t,k}^I + \mbf w_t)^H ], \notag \\
        % & = \sum_{k,l} \rho_{t,k,l}^I \mbb E \left[ (\mbf b  \otimes \mbf s_{t,k,l}) (\mbf b  \otimes \mbf s_{t,k,l})^H \right] + \gamma P^{-1}\mbf I  \notag \\
        & = \sum_{k,l} \rho_{t,k,l}^I  \left( \mbf b  \mbf b ^H \right) \otimes \mbb E [\mbf s_{t,k,l} \mbf s_{t,k,l}^H] + \gamma P^{-1}\mbf I, \label{Egg_2} \\
        % & = (\gamma P^{-1} +  P^I P^{-1} \sum_{k,l} \rho_{t,k,l}^I \mbf b  \mbf b ^H) \otimes \mbf I, \label{Egg_3} \\
        & = \mbf Q_{t} \otimes \mbf I, \label{R_t^I}
    \end{align} 
    \end{subequations}
where $\mbf Q_{t} = (\gamma P^{-1} \mbf I +  P^I P^{-1} \sum_{k,l} \rho_{t,k,l}^I \mbf b  \mbf b ^H)$, \eqref{Egg_2} follows from \eqref{45} and $ (\mbf A_1 \otimes \mbf A_2)(\mbf B_1 \otimes \mbf B_2) = (\mbf A_1 \mbf B_1) \otimes (\mbf A_2 \mbf B_2) $; \eqref{R_t^I} follows from $\mbb E[ \mbf s_{t,k,l} \mbf s_{t,k,l}^H ] = P^I P^{-1} \mbf I$. 
%The covariance of  $\mrm{vec}(\mbf F^H \mbf X^{-1} \mbf X_{t,k}^I \mbf H_{t,k}^I) $ equals \eqref{Egg_3} by noting that $\tilde{\mbf s}_{t,k,l} = \mbf F^H \mbf X^{-1} \mbf X_{t,k}^I \mbf a_N(\tau_{t,k,l}^I)$ has the property of $\mbb E[\tilde{\mbf s}_{t,k,l} \tilde{\mbf s}_{t,k,l}^H] = P^I \mbf I$. 
Note that $\mbf Q_{t}$ in \eqref{R_t^I} represents the spatial covariance of $\sum_k \tilde{\mbf h}_{t,k}^I + \mbf w_t$.
% due to unknown $\rho_{t,k,l}^I$, $\theta_{t,k,l}^I$, and $\phi_{t,k,l}^I$ 
We adopt a sample covariance method to estimate $\mbf Q_{t}$. Specifically, we transform the frequency-space-domain signal $\mbf r_t$ into the delay-space-domain as 
\begin{align}
 \tilde{\mbf r}_t = (\mbf I \otimes \mbf F^H) \mbf r_t = \tilde{\mbf h}_t +  (\mbf I \otimes \mbf F^H) \Big(\sum_k \tilde{\mbf h}_{t,k}^I + \mbf w_t \Big),
\end{align}
where $\tilde{\mbf h}_t = (\mbf I \otimes \mbf F^H) \mbf h_t$ represents the  delay-space-domain channel with $\mbf F$ being the DFT matrix. The maximum delay spread of $\tilde{\mbf h}_t$ is typically much smaller than the OFDM symbol duration, i.e., $\tilde{h}_{t, mN+j} \approx 0$ for $j > J_0$ and any $m \in\{0,...,M-1\}$. Then, the signal element $\tilde{r}_{t,mN+j},j > J_0,\forall m,$ includes only the interference $(\mbf I \otimes \mbf F^H)(\sum_k \tilde{\mbf h}_{t,k}^I + \mbf w_t)$. Furthermore, given $\mbf R_t^I = \mbf Q_{t} \otimes \mbf I$, the covariance of $(\mbf I \otimes \mbf F^H)(\sum_k \tilde{\mbf h}_{t,k}^I + \mbf w_t)$ equal $(\mbf I \otimes \mbf F^H)(\mbf Q_{t} \otimes \mbf I)(\mbf I \otimes \mbf F) = \mbf Q_{t} \otimes \mbf I$. This result indicates that the covariance of $\sum_k \tilde{\mbf h}_{t,k}^I + \mbf w_t$ remains unchanged under the transformation by $\mbf I \otimes \mbf F^H$. As such, we next determine the maximum delay spread of $\tilde{\mbf h}_t$ and use the sample covariance of $\tilde{\mbf r}_t$ to estimate $\mbf Q_t$. Given the CKM outputs $\{\bar{\rho}_{q,l},\bar{\tau}_{q,l},\bar{\theta}_{q,l},\bar{\phi}_{q,l}\}$, the PDP of $\tilde{\mbf h}_t$ is expressed as 
\begin{align}
 \mrm{P}_{q,j} & = \frac{1}{M} \sum_{m=0}^{M-1} \mbb E_{\{\bar{\alpha}_{t,l}\}}[  |\tilde{\mbf h}_{t,mN+j}|^2 ] \notag \\
    & =  \sum_l \bar{\rho}_{q,l} \frac{||\mbf b_l||_2^2}{M}  ||\mbf F^H \mbf a_N(\bar{\tau}_{q,l})||_2^2, ~ j=1,...,N.
\end{align} 
Denote by $P^{\mrm{thres}}$ the power threshold and define $\mal J = \{ j | \mrm{P}_{q,j} < P^{\mrm{thres}} \}$ to represent the set of indices of delay taps corresponding to $\tilde{\mbf h}_{t,mN+j} \approx 0$. Then, we estimate $\mbf Q_{t}$ as
\begin{align}
 \hat{\mbf Q}_{t} = \frac{1}{|\mal{J}|} \sum_{j \in \mal{J}}  & [\tilde{r}_{t,0N+j},...,\tilde{r}_{t,(M-1)N+j}]^T \notag \\
 &\times [\tilde{r}_{t,0N+j},...,\tilde{r}_{t,(M-1)N+j}]^*. \label{C_I_hat}
\end{align}
With \eqref{R_t^I} and \eqref{C_I_hat}, we obtain $\hat{\mbf R}_t^I = \hat{\mbf Q}_{t} \otimes \mbf I$.

% Note that

% $\rho_{t,k,l}^I,\theta_{t,k,l}^I,\phi_{t,k,l}^I$
% We apply the IDFT to $\mbf X^{-1} \mbf Y_t$

% power delay profile (PDP) of $\mbf F^H \mbf H_t$ 

% $ \mbf p_q =  \sum_l \bar{\rho}_{q,l} \frac{|| \mbf b ||}{M} \mbf F^H \mbf a_N(\bar{\tau}_{q,l}) $

% $\hat{\mbf R}^I = \mbf I \otimes (\sum_{j \in \mal J} \mbf{u}_{t,j} \mbf{u}_{t,j})$  $\mal J = \{ j | p_{q,j} < P^{\mrm{thres}} \}$

\vspace{-0.2cm}
\subsection{Low-Complexity MMSE-IRC}

The computational complexity of MMSE-IRC in \eqref{IRC} is dominated by the matrix inversion $( \hat{\mbf C}_{t} + \hat{\mbf R}_t^I)^{-1}$ with complexity $\mal{O}(N^3M^3)$, which is  prohibitively high for relatively large $N$ and $M$. To reduce the computational complexity, we introduce a matrix factorization of $\hat{\mbf C}_{t}$ as 
\begin{align}
 \hat{\mbf C}_{t} = \mbf A_{q} \bsm{\Sigma}_q \mbf A_{q}^H, \label{fac}
\end{align}
where $\mbf A_{q} \! = \! [\mbf b(\bar{\theta}_{q,1},\bar{\phi}_{q,1}) \otimes {\mbf a}_{N} (\bar{\tau}_{q,1}),...,\mbf b(\bar{\theta}_{q,\bar{L}},\bar{\phi}_{q,\bar{L}})\otimes{\mbf a}_{N} (\bar{\tau}_{q,\bar{L}})] \in \mbb C^{NM\times \bar{L}}$ and $\bsm{\Sigma}_q \! = \! \mrm{diag}([\bar{\rho}_{q,1},...,\bar{\rho}_{q,\bar{L}}]^T)$. By substituting \eqref{fac} and $\hat{\mbf R}_t^I = \hat{\mbf Q}_{t} \otimes \mbf I$ into \eqref{IRC}, we obtain
\begin{subequations}
    \begin{align}
 \hat{\mbf h}_t & = \mbf A_{q} \bsm{\Sigma}_q \mbf A_{q}^H  \left( \mbf A_{q} \bsm{\Sigma}_q \mbf A_{q}^H + \hat{\mbf Q}_{t} \otimes \mbf I \right)^{-1} \mbf r_t, \label{42b} \\
        & = \mbf A_{q} \left( \mbf A_{q}^H (\hat{\mbf Q}_{t}^{-1} \otimes \mbf I) \mbf A_{q} + \bsm{\Sigma}_q^{-1} \right)^{-1} \notag \\
        & \hspace{0.4cm} \times \mbf A_{q}^H (\hat{\mbf Q}_{t}^{-1} \otimes \mbf I) \mbf r_t, \label{42c}
    \end{align}
\end{subequations}
where \eqref{42c} is from the Woodbury matrix identity. With \eqref{42c}, the computational complexity of the matrix inversion reduces from $\mal{O}(N^3M^3)$ to $\mal{O}(\bar{L}^3 + M^3)$, where $\bar{L}$ can be modified to strike a complexity-performance balance. In \eqref{42c}, $\mbf A_{q}^H (\hat{\mbf Q}_{t}^{-1} \otimes \mbf I) \mbf A_{q}$ dominates the complexity with $\bar{L}^2NM+\bar{L}M^2N$ complex multiplications. We next show that the complexity of $\mbf A_{q}^H (\hat{\mbf Q}_{t}^{-1} \otimes \mbf I) \mbf A_{q}$ can be substantially reduced by exploiting the properties of the Khatri-Rao product. Specifically, denote by $\circledast_c$ the column-wise Khatri-Rao product. We express $\mbf A_{q}$ as $\mbf A_{q} = \mbf B_{q} \circledast_c \mbf A_{f,q}$ with $\mbf A_{f,q} = [{\mbf a}_{N} (\bar{\tau}_{q,1}),...,{\mbf a}_{N} (\bar{\tau}_{q,\bar{L}})]$ and $\mbf B_{q} = [\mbf b(\bar{\theta}_{q,1},\bar{\phi}_{q,1}),...,\mbf b(\bar{\theta}_{q,\bar{L}},\bar{\phi}_{q,\bar{L}})]$. Then, we obtain 
\begin{subequations}
    \begin{align}
        & \mbf A_{q}^H (\hat{\mbf Q}_{t}^{-1} \otimes \mbf I) \mbf A_{q} \notag \\
 = & (\mbf B_q \circledast_c \mbf A_{f,q})^H (\hat{\mbf Q}_{t}^{-1} \otimes \mbf I) (\mbf B_q \circledast_c \mbf A_{f,q}), \label{Kro_1} \\
 = & (\mbf B_q^H \circledast_r \mbf A_{f,q}^H) ((\hat{\mbf Q}_{t}^{-1} \mbf B_q) \circledast_c \mbf A_{f,q}), \label{Kro_2} \\
 = & (\mbf B_q^H \hat{\mbf Q}_{t}^{-1} \mbf B_q) \odot (\mbf A_{f,q}^H\mbf A_{f,q}) \label{Kro_3},
    \end{align} 
\end{subequations} 
where \eqref{Kro_2} is from $(\mbf A_1 \otimes \mbf A_2)(\mbf A_3 \circledast_c \mbf A_4) = (\mbf A_1 \mbf A_2) \circledast_c (\mbf A_3 \mbf A_4) $; \eqref{Kro_3} is from $(\mbf A_1 \circledast_r \mbf A_2)(\mbf A_3 \circledast_c \mbf A_4) = (\mbf A_1 \mbf A_2) \odot (\mbf A_3 \mbf A_4)$ with $\circledast_r$ representing the row-wise Khatri-Rao product. With \eqref{Kro_3}, the computational complexity of $\mbf A_{q}^H (\hat{\mbf Q}_{t}^{-1} \otimes \mbf I) \mbf A_{q}$ is reduced from $\bar{L}^2NM+\bar{L}M^2N$ to $\bar{L}^2(N+M) + \bar{L}M^2$. Overall, we reduce the computational of MMSE-IRC from $\mal{O}(N^3 M^3)$ to $\mal{O}(M^2(M+\bar{L})+(M+N)\bar{L}^2+\bar{L}^3)$.

\subsection{Error Bound of CKM Representation}

In this subsection, we analyze how the CKM outputs impacts the estimation performance of the proposed MMSE-IRC. In \eqref{42c}, the MMSE-IRC can be treated as a linear transformation of $\hat{\mbf h}_t$ with the transformation matrix $ \mbf A_{q} $, i.e.,
\begin{align}
	\hat{\mbf h}_t = \mbf A_{q} \hat{\bsm \beta}_t,
\end{align}
where $\hat{\bsm \beta}_t = \left( \mbf A_{q}^H (\hat{\mbf Q}_{t}^{-1} \otimes \mbf I) \mbf A_{q} + \bsm{\Sigma}_q^{-1} \right)^{-1} \mbf A_{q}^H (\hat{\mbf Q}_{t}^{-1} \otimes \mbf I) \mbf r_t$. Clearly, the transformation matrix $\mbf A_{q}$ affects the reconstruction accuracy $||\hat{\mbf h}_t - \mbf h_t ||_2^2$. We refer to the accuracy lower-bound of this linear reconstruction as \emph{CKM accuracy}:

\begin{definition}
	 The CKM accuracy in time-slot $t$ is defined as
   \begin{align}
	\mal{L}_{\mrm{CKM}} = \min_{\bsm{\beta}_{t}} \frac{1}{N N_b} || \mbf h_t - \mbf A_{q} \bsm{\beta}_{t} ||_2^2.          
   \end{align}
\end{definition}
We next show that the MMSE-IRC performance converges to $\mal{L}_{\mrm{CKM}}$ in the interference-free and error-free case. 

\begin{proposition} 
	In the interference-free and error-free case of system model \eqref{r_t}, i.e., $\tilde{\mbf h}_{t,k}^I=0,\forall k$ and $\mbf w_t=0$, the estimation MSE of the MMES-IRC in \eqref{42c} equals $\mal{L}_{\mrm{CKM}}$.
\end{proposition}

\begin{proof}
	$\mal{L}_{\mrm{CKM}}$ is obtained when $\hat{\bsm \beta}_t$ takes the least-square (LS) estimate as $\left( \mbf A_{q}^H \mbf A_{q} \right)^{-1} \mbf A_{q}^H \mbf h_t$. The corresponding estimate of $\mbf h_t$ is
	\begin{align}
		\mbf h_t^* = \mbf A_{q} \hat{\bsm \beta} =  \mbf A_{q} \left( \mbf A_{q}^H \mbf A_{q} \right)^{-1} \mbf A_{q}^H \mbf h_t. \label{g_LS}
	\end{align}
	By adopting the SVD decomposition $\mbf A_{q} =\mbf U \mbf D \mbf V^H$, $\mbf h_t^*$ in \eqref{g_LS} is equivalently expressed as 
	\begin{align}
		\mbf h_t^* = \mbf A_{q} \mbf V \mbf D^{\dagger}  \mbf U^H \mbf h_t.
	\end{align}
	With $\mbf A_{q} = \mbf U \mbf D \mbf V^H$, $\tilde{\mbf h}_{t,k}^I=0,\forall k$ and $\mbf w_t=0$, the MMSE-IRC \eqref{42b} is expressed as 
	\begin{subequations}
		\begin{align}
			 \hat{\mbf h}_t 
			= & \mbf A_{q} \bsm{\Sigma}_q \mbf A_{q}^H  \left( \mbf A_{q} \bsm{\Sigma}_q \mbf A_{q}^H \right)^{\dagger} \mbf r_t, 
			\label{lim_a}  \\
			= & \mbf A_{q} \bsm{\Sigma}_q \mbf V  \mbf D^H (\mbf D \mbf V^H \bsm{\Sigma}_q \mbf V \mbf D^H)^{\dagger} \mbf U^H \mbf h_t, \label{lim_b} \\
			= & \mbf A_{q} \mbf V  \mbf D^{\dagger} \mbf U^H \mbf h_t.
		\end{align}
	\end{subequations}
	In the following section, we show $\mal{L}_{\mrm{CKM}}$ as a performance lower bound of the MMSE-IRC scheme. 
	% where \eqref{lim_a} and \eqref{lim_b} follow from $\hat{\mbf R}^I = \mbf 0$ and $\hat{\mbf r}_t = \mbf h_t$, respectively, in the interference-free and error-free case.
	% By noting that $(\mbf D^H \mbf D)^{-1} \mbf D = \mbf D^H (\mbf D \mbf D^H)^{-1}$, we obtain $\lim_{P \to \infty} \hat{\mbf g}^{\mrm{3D}} = \mbf h_t^*$, which completes the proof.
	
	% Represent the channel matrix $\mbf{g}_{\mrm{3D}}$ as a linear combination of steering vectors. The received signal can then be expressed as
	% \begin{align}
	% \tilde{\mbf y} = \mathbf A_{\mrm{3D}} \bsm{\gamma} + \mbf w,
	% \end{align}
	% where $\mbf w \sim \mal{CN}(0,\sigma_w^2)$ is the model error plus white noise term, and $\bsm{\gamma}$ is the small-scale fading vector. The MMSE estimator of the small-scale fading coefficients is given by
	% \begin{align}
	% 	\hat{\bsm{\gamma}} = \left( (\mbf A_{q})^H \mbf A_{q} + \frac{\sigma_w^2}{P} \mbf I \right)^{-1} (\mbf A_{q})^H \tilde{\mbf y}.
	% \end{align}
	% In the high SINR regime $P \to \infty $, the above equation converges to 
	% \begin{align}
	% 	\hat{\bsm{\gamma}} = \left( (\mbf A_{q})^H \mbf A_{q} \right)^{-1} (\mbf A_{q})^H \tilde{\mbf y},
	% \end{align}      
	% which completes the proof.
\end{proof}

% Warm-start to calibrate the path parameters due to drifting and birth-death.
\vspace{-0.7cm}

\section{Numerical Results}

In this section, we evaluate the performance of the CKM construction and the CKM-assisted channel estimation. Consider that the BS is equipped with a $4 \times 8$ URA and serves a user on $N=192$ subcarriers. The subcarrier spacing is $\Delta_f =30$ KHz, and the carrier frequency is $6.5$ GHz. We generate the channel response of the user by following the standard TR 38.901 in the UMa NLOS scenario \cite[Table 7.5-6]{TR38901}, where the number of channel clusters is set to $10$ with $20$ sub-paths per cluster, i.e., $L=200$. Besides, the channel parameters change at different positions following the spatial consistency procedure \cite[Sec. 7.6.3.2]{TR38901}. 
% The height of the user is $1.5$ m, and the initial location is randomly generated within the region $[100,500]$ m. 
In the $q$-th grid, we generate the channel dataset of the user by uniformly sampling a length-$d$ line $|\mbb T_q|=10 d$ times. Then, the channel dataset is randomly divided into two sub-sets with the ratio being $7:3$, where the former is used for CKM construction, and the latter is used for channel estimation. The number of interferers is $4$ with the sparsity being $0.5$. The channel models of interferences also follow the standard TR 38.901 in the UMa NLOS scenario with $L_k^I=200,\forall k$.
% The user and iterferers move along a line in a random direction. 
In CKM construction, the iteration number of Algorithm 1 is set to $40$. The parameters of $\bar{L}_k^I=15$ channel paths are extracted per interferer for interference cancellation. In channel estimation, the interference covariance is estimated based on \eqref{C_I_hat}, where $P^{\mrm{thres}}$ takes the value which is $10$ dB lower than the power of interferences plus noise. 
We adopt the constant-modulus pilot with phase randomly taking within $[0,2\pi]$. To evaluate the performance of the CKM construction, we consider the baselines as follows:
\begin{itemize}
    \item OMP-based CKM: The OMP algorithm \cite{OMP} is a greedy algorithm which sequentially finds the first $\bar{L}$ best matching $\bar{\tau}_{q,l}$, $\bar{\theta}_{q,l}$ and $\bar{\phi}_{q,l}$ based on a delay-angular-domain dictionary generated from a uniform sampling over $\mbf a_N(\cdot) \otimes \mbf a_{M_1}(\cdot) \otimes \mbf a_{M_2}(\cdot)$. 
    %Note that $P$ is typically unknown in OMP. OMP with small $P$ will ignore some dominant channel paths, and OMP with large $P$ may extract noise components, causing error propagation. Through parameter tuning, we select $P=50$ as a robust value. 
 After obtaining $\{\bar{\tau}_{q,l},\bar{\theta}_{q,l},\bar{\phi}_{q,l}\}_{l=1}^{\bar{L}}$, $\bsm{\beta}_t$ is estimated by the LS based on \eqref{y_for_beta}. Then, the path powers are obtained as $\bar{\rho}_{q,l} = \sum_t |\hat{\beta}_{t,l}|^2 / |\mbb{T}_q|,\forall l$.  % Note that the OMP algorithm requires an accurate estimate on sparsity level $P$.
    \item ICI-non-cognitive CKM: This scheme is used to evaluate the effectiveness of the interference cancellation of Algorithm 1. Specifically, we treat the interference signals as AWGN and remove the parameter estimation of interferences in Algorithm 1.  
\end{itemize}

To evaluate the performance of channel estimation, we consider the state-of-the-art as follows:
\begin{itemize}
    \item OMP: the steps follow those in OMP-based CKM, except that the computation of $\{\bar{\rho}_{q,l}\}$ is removed.
    \item TMP: The TMP algorithm \cite{STCS} is a compressed sensing method for sparse signal recovery. The original TMP comprises an LMMSE module and a denoiser designed for an AWGN signal model. We modify the LMMSE module to exploit the estimate of the interference covariance using \eqref{C_I_hat} with $j \in[N/4+1,...,3N/4]$. 
    \item VBI-MMSE: The VBI-MMSE \cite{Vincent} adopts the variational Bayesian inference to learn the PDP and spatial covariance of the user channel. By using this information together with an estimate of interference spatial covariance, an MMSE estimator is adopted separately for user spatial channel per delay tap.
    \item Lower bound (LB): The LB corresponds to the CKM accuracy.
\end{itemize}

% 1-2 from one code 

% 1. CKM accuracy versus L = $10/20/30/40/50$ (SINR = $-20:5:20$ dB); add method w/o interference-cancellation

% \subsection{Results with Channel Model of TR 38.901}

\begin{figure}[h] 
	\centering
	\includegraphics[width = 3.2 in]{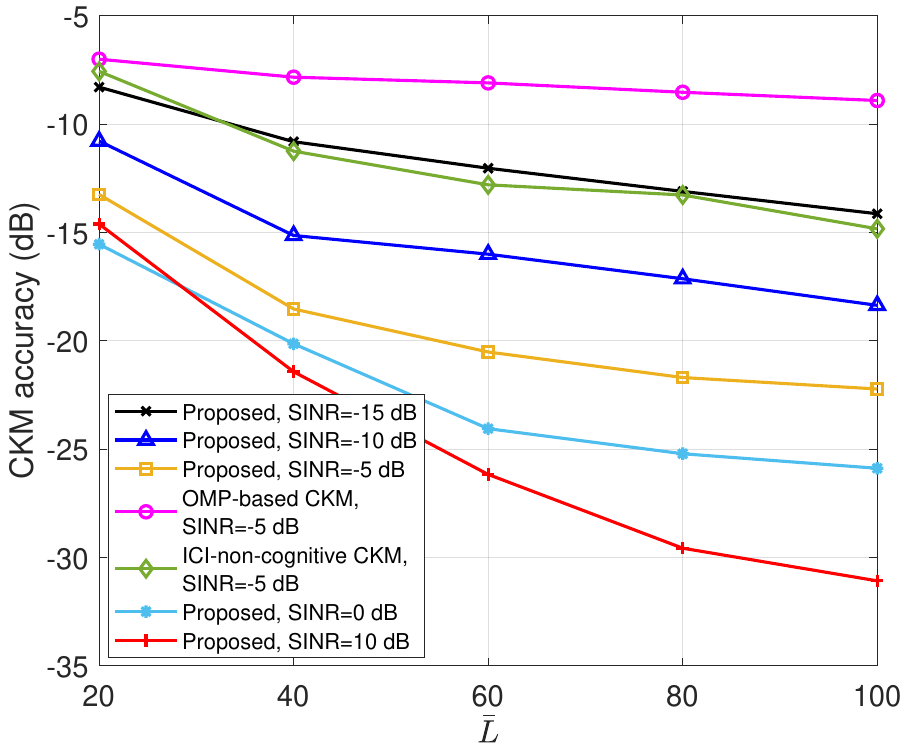}
	\caption{CKM accuracy versus $\bar{L}$. $d = 2$ m.}
	\label{Fig1}
\end{figure}

Fig. \ref{Fig1} shows the CKM accuracy versus $\bar{L}$ at different SINRs. The CKM accuracy of the proposed scheme improves with the increase of $\bar{L}$ and $\mrm{SINR}$. Particularly, a CKM accuracy of $-20.5$ dB is achieved at $\bar{L}=60$ and $\mrm{SINR}= -5$ dB. As comparisons, the accuracies of OMP-based and ICI-non-cognitive CKMs are $-8$ dB and $-12.8$ dB, respectively, which are $12.5$ dB and $7.7$ dB higher than that of Algorithm 1. This result shows a high accuracy of Algorithm 1 at relatively low SINRs and validates the effectiveness of the interference cancellation of Algorithm 1. 

% 2. MSE of channel estimation versus SINR (L = $5/10/25/50$); CKM accuracy as LB

\begin{figure}[h] 
	\centering
	\includegraphics[width = 3.2 in]{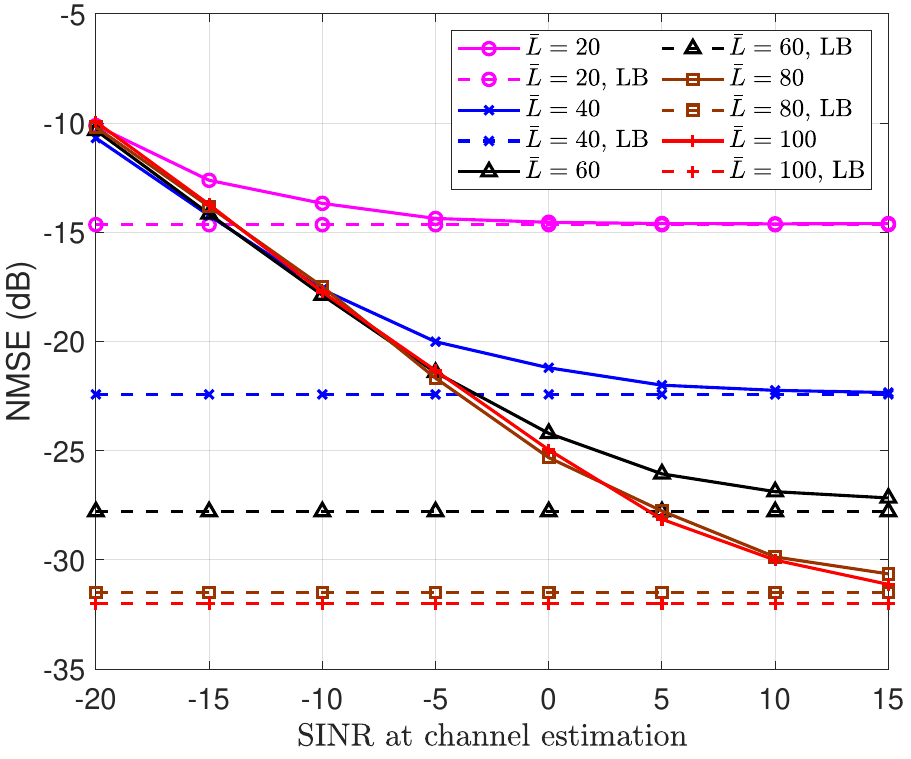}
	\caption{The MSE of the proposed channel estimation versus SINR at different $\bar{L}$. $d$ = 2 m and SINR $= 10$ dB in CKM construction.}
	\label{Fig2}
\end{figure}

% 3. CKM accuracy versus d = $1/2/3/4/5$  (SINR = x/x/x) (optional); add method w/o interference-cancellation

In Fig. \ref{Fig2}, we present the MSE performance of channel estimation for the proposed MMSE-IRC. The estimation performance improves as the SINR increases, and gradually approaches the LB. The MSE performance with a larger $\bar{L}$ outperforms that with a smaller $\bar{L}$. On the other hand, a larger $\bar{L}$ results in a higher computational complexity in channel estimation. In practice, we should choose an appropriate $\bar{L}$ to strike a balance between the estimation performance and the computational complexity.

% \begin{figure}[h] 
% 	\centering
% 	\includegraphics[width = 3.5 in]{NMSE_vs_SINR_d.pdf}
% 	\caption{MSE of channel estimation versus SINR at different $d$ for the proposed method. $\bar{L}=60$ and SINR $= 10$ dB in CKM construction.}
% 	\label{Fig3}
% \end{figure}

% 4. MSE of channel estimation versus SINR under different $d$; CKM accuracy as LB

% Fig. \ref{Fig3} evaluates how the channel estimation performance varies when using the CKM constructed from different sizes of location grid. We observe that when $d$ increases, there is a degradation in CKM accuracy. This is because a larger physical area corresponds to a larger drifting of channel parameters, which leads to a higher channel representation error. The performances of channel estimation at different $d$ are similar when the $\mrm{SINR}$ is relatively low , i.e., $\mrm{SINR}\le-10$ dB, but the performance gap becomes significant as the $\mrm{SINR}$ increases. 

\begin{figure}[h] 
	\centering
	\includegraphics[width = 3.2 in]{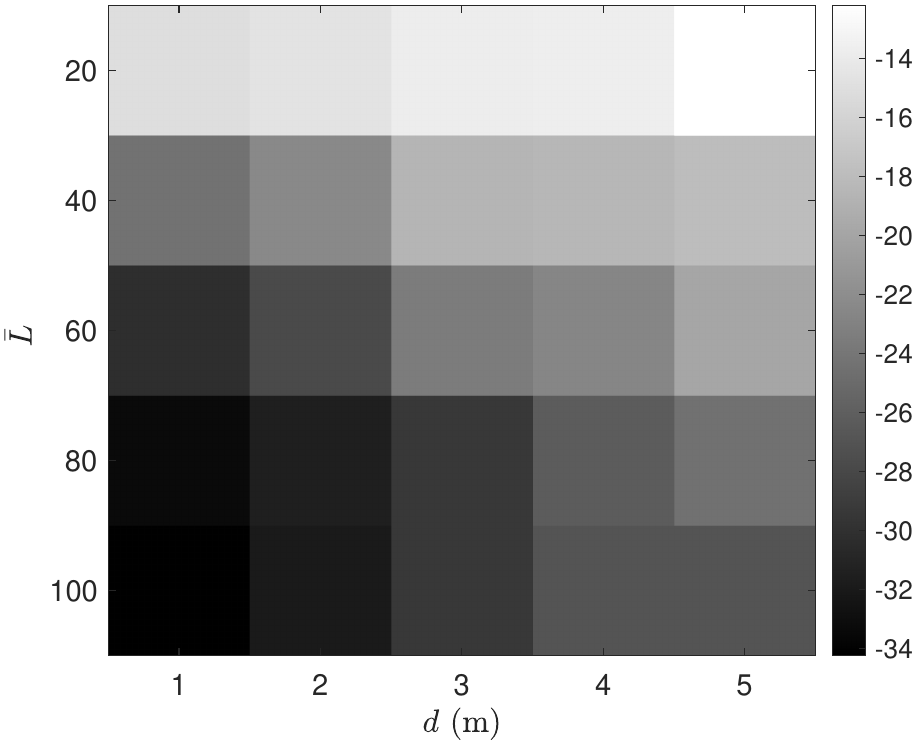}
	\caption{The CKM accuracy of Algorithm 1 at different $\bar{L}$ and $d$. SINR $=10$ dB in CKM construction.}
	\label{Fig4}
\end{figure}

% 5. Algorithm Comparisons: MSE of channel estimation versus SINR; add method w/o interference-cancellation and pilot reduction

Fig. \ref{Fig4} shows the CKM accuracy of Algorithm 1 as a function of $d$ and $\bar{L}$. It is seen that the increase of $\bar{L}$ or the decrease of $d$ corresponds to the enhancement of the CKM accuracy. For example, to achieve the CKM accuracy of $-20$ dB, we require $\bar{L}\ge40$ at $d=2$ or $\bar{L}\ge60$ at $d=3$. This observation indicates that when $d$ is limited to a specific value due to practical restrictions, e.g., positioning accuracy, we can choose an appropriate $\bar{L}$ to achieve the desired CKM accuracy.

\begin{figure}[h] 
	\centering
	\includegraphics[width = 3.1 in]{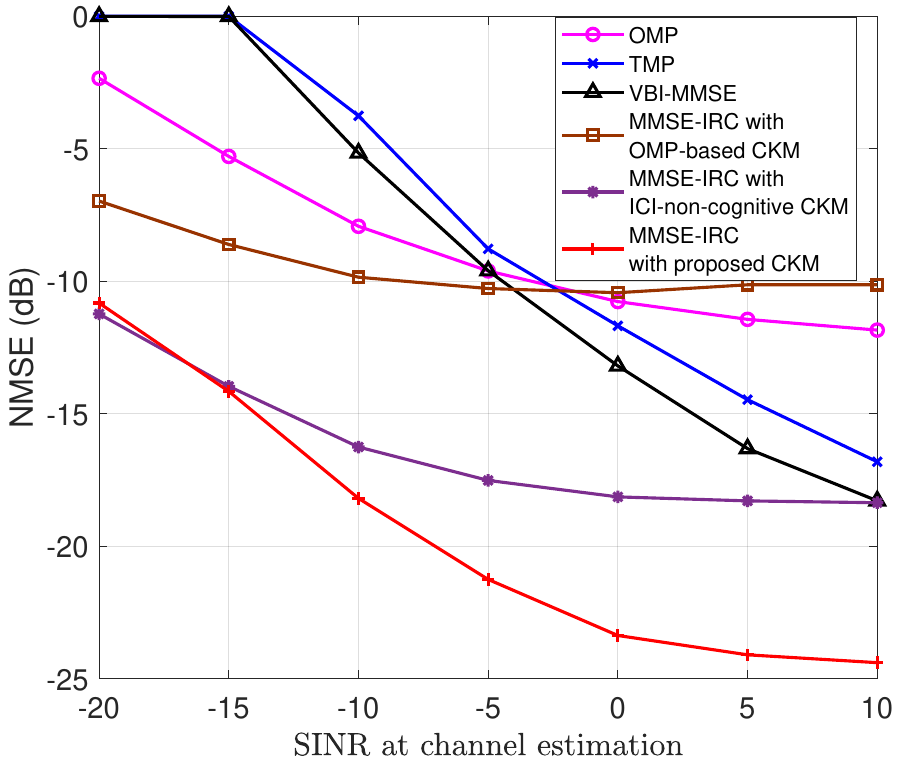}
	\caption{The MSE of channel estimation versus SINR. $\bar{L}=60$, $d=2$ m, and SINR $=0$ dB in CKM construction.}
	\label{Fig5}
\end{figure}

Fig. \ref{Fig5} provides the channel estimation performances of the proposed MMSE-IRC and the baselines. We see that OMP outperforms TMP and VBI-MMSE at $\mrm{SINR}<-5$ dB but is outperformed by TMP and VBI-MMSE at $\mrm{SINR}>-5$ dB. This is because TMP and VBI-MMSE are designed in the Bayesian framework, and the prior learning is inaccurate at a relatively low SINR. The proposed scheme significantly outperforms all the baselines. Particularly, to achieve $\mrm{MSE}=-15$ dB, TMP and VBI-MMSE require $\mrm{SINR}\ge5$ dB and $\mrm{SINR}\ge2.5$ dB, respectively, while the proposed scheme only needs $\mrm{SINR}\ge-15$ dB with an $\mrm{SINR}$ reduction of at least $17.5$ dB. Furthermore, the proposed MMSE-IRC with proposed CKM achieves $\mrm{NMSE}\approx -25$ dB at $\mrm{SINR}=5$ dB, while the MMSE-IRC with OMP-based and ICI-non-cognitive CKM reach $\mrm{NMSE}\approx -18$ dB and $\mrm{NMSE}\approx -10$ dB, respectively. These results indicate that the proposed CKM construction is able to considerably enhance the channel estimation.

% which demonstrates the superiority of the proposed CKM construction for channel estimation enhancement. 

\vspace{-0.4cm}

\section{Conclusions}

In this paper, we investigated the CKM construction and the CKM-assisted channel estimation in MIMO-OFDM systems with ICI. We divided the physical region into grids and used the BS received signals per grid to extract channel model parameters. The CKM construction problem was formulated into a Bayesian inference problem. Then, a hybrid message-passing algorithm was developed, integrating the mechanism of interference cancellation. Based on the CKM outputs, we designed the CKM-assisted MMSE-IRC channel estimator, where the computational complexity was substantially reduced by exploiting the structure of the user channel covariance. Numerical results verified the high accuracy of the proposed CKM and showed the superiority of the CKM-assisted channel estimator over the state-of-the-art counterparts. 

\vspace{-0.3cm} 

\begin{appendices}

\end{appendices}

\bibliographystyle{IEEEtran}
\bibliography{TurboMP}

% Generated by IEEEtran.bst, version: 1.14 (2015/08/26)
\begin{thebibliography}{10}
\providecommand{\url}[1]{#1}
\csname url@samestyle\endcsname
\providecommand{\newblock}{\relax}
\providecommand{\bibinfo}[2]{#2}
\providecommand{\BIBentrySTDinterwordspacing}{\spaceskip=0pt\relax}
\providecommand{\BIBentryALTinterwordstretchfactor}{4}
\providecommand{\BIBentryALTinterwordspacing}{\spaceskip=\fontdimen2\font plus
\BIBentryALTinterwordstretchfactor\fontdimen3\font minus \fontdimen4\font\relax}
\providecommand{\BIBforeignlanguage}[2]{{%
\expandafter\ifx\csname l@#1\endcsname\relax
\typeout{** WARNING: IEEEtran.bst: No hyphenation pattern has been}%
\typeout{** loaded for the language `#1'. Using the pattern for}%
\typeout{** the default language instead.}%
\else
\language=\csname l@#1\endcsname
\fi
#2}}
\providecommand{\BIBdecl}{\relax}
\BIBdecl

\bibitem{Zhengquan}
Z.~Zhang, Y.~Xiao, Z.~Ma, M.~Xiao, Z.~Ding, X.~Lei, G.~K. Karagiannidis, and P.~Fan, ``{6G} wireless networks: Vision, requirements, architecture, and key technologies,'' \emph{IEEE Veh. Technol. Mag.}, vol.~14, no.~3, pp. 28--41, Sep. 2019.

\bibitem{you2021towards}
X.~You, C.-X. Wang, J.~Huang, X.~Gao, Z.~Zhang, M.~Wang, Y.~Huang, C.~Zhang, Y.~Jiang, J.~Wang \emph{et~al.}, ``Towards {6G} wireless communication networks: Vision, enabling technologies, and new paradigm shifts,'' \emph{Sci. China Inf. Sci.}, vol.~64, pp. 1--74, Jan. 2021.

\bibitem{XLMIMO}
J.~Wang, C.-X. Wang, J.~Huang, H.~Wang, and X.~Gao, ``A general {3D} space-time-frequency non-stationary {THz} channel model for {6G} ultra-massive {MIMO} wireless communication systems,'' \emph{IEEE J. Sel. Areas Commun.}, vol.~39, no.~6, pp. 1576--1589, June 2021.

\bibitem{RME}
D.~Romero and S.-J. Kim, ``Radio map estimation: A data-driven approach to spectrum cartography,'' \emph{IEEE Signal Process. Mag.}, vol.~39, no.~6, pp. 53--72, Nov 2022.

\bibitem{Ray}
D.~He, B.~Ai, K.~Guan, L.~Wang, Z.~Zhong, and T.~Kürner, ``The design and applications of high-performance ray-tracing simulation platform for {5G} and beyond wireless communications: A tutorial,'' \emph{IEEE Commun. Surv. Tutor.}, vol.~21, no.~1, pp. 10--27, Firstquarter 2019.

\bibitem{JianWen}
Y.~Sun, J.~Zhang, Y.~Zhang, L.~Yu, Z.~Yuan, G.~Liu, and Q.~Wang, ``Environment features-based model for path loss prediction,'' \emph{IEEE Wireless Commun. Lett.}, vol.~11, no.~9, pp. 2010--2014, Sep. 2022.

\bibitem{Int}
A.~B.~H. Alaya-Feki, S.~B. Jemaa, B.~Sayrac, P.~Houze, and E.~Moulines, ``Informed spectrum usage in cognitive radio networks: Interference cartography,'' in \emph{Proc. IEEE Int. Symp. Personal, Indoor Mobile Radio Commun.}, 2008, pp. 1--5.

\bibitem{PSD}
D.~Romero, S.-J. Kim, G.~B. Giannakis, and R.~López-Valcarce, ``Learning power spectrum maps from quantized power measurements,'' \emph{IEEE Trans. Signal Process.}, vol.~65, no.~10, pp. 2547--2560, May 2017.

\bibitem{CKM2}
Y.~Zeng, J.~Chen, J.~Xu, D.~Wu, X.~Xu, S.~Jin, X.~Gao, D.~Gesbert, S.~Cui, and R.~Zhang, ``A tutorial on environment-aware communications via channel knowledge map for {6G},'' \emph{IEEE Commun. Surv. Tutor.}, pp. 1--1, Feb 2024.

\bibitem{ISAC}
F.~Liu, Y.~Cui, C.~Masouros, J.~Xu, T.~X. Han, Y.~C. Eldar, and S.~Buzzi, ``Integrated sensing and communications: Toward dual-functional wireless networks for {6G} and beyond,'' \emph{IEEE J. Sel. Areas Commun.}, vol.~40, no.~6, pp. 1728--1767, June 2022.

\bibitem{BF_CKM}
D.~Wu, Y.~Zeng, S.~Jin, and R.~Zhang, ``Environment-aware hybrid beamforming by leveraging channel knowledge map,'' \emph{IEEE Trans. Wireless Commun.}, vol.~23, no.~5, pp. 4990--5005, May 2024.

\bibitem{Lol_CKM}
Y.~Long, Y.~Zeng, X.~Xu, and Y.~Huang, ``Environment-aware wireless localization enabled by channel knowledge map,'' in \emph{IEEE GLOBECOM}, 2022, pp. 5354--5359.

\bibitem{CKM1}
Y.~Zeng and X.~Xu, ``Toward environment-aware {6G} communications via channel knowledge map,'' \emph{IEEE Wireless Commun.}, vol.~28, no.~3, pp. 84--91, June 2021.

\bibitem{EM_CKM}
K.~Li, P.~Li, Y.~Zeng, and J.~Xu, ``Channel knowledge map for environment-aware communications: {EM} algorithm for map construction,'' in \emph{IEEE WCNC}, 2022, pp. 1659--1664.

\bibitem{Chl_sounder}
P.~B. Papazian, C.~Gentile, K.~A. Remley, J.~Senic, and N.~Golmie, ``A radio channel sounder for mobile millimeter-wave communications: System implementation and measurement assessment,'' \emph{IEEE Trans. on Microw. Theory and Tech.}, vol.~64, no.~9, pp. 2924--2932, Sep. 2016.

\bibitem{TS38213}
\emph{NR Physical Layer Procedures for Control}, Std. TS 38.213 Release 17, 3GPP, June. 2022.

\bibitem{tse2005fundamentals}
D.~Tse and P.~Viswanath, \emph{Fundamentals of {Wireless} {Communication}}.\hskip 1em plus 0.5em minus 0.4em\relax Cambridge {University} {Press}, 2005.

\bibitem{TR38901}
\emph{{5G}; Study on {Channel} {Model} for {Frequencies} from 0.5 to 100 {GHz}}, Std. TR 38.901 Release 16, 3GPP, Nov. 2020.

\bibitem{BP}
J.~Yedidia, W.~Freeman, and Y.~Weiss, ``Constructing free-energy approximations and generalized belief propagation algorithms,'' \emph{IEEE Trans. Inf. Theory}, vol.~51, no.~7, pp. 2282--2312, June 2005.

\bibitem{Kay}
S.~M. Kay, \emph{Fundamentals of Statistical Signal Processing: Estimation Theory}.\hskip 1em plus 0.5em minus 0.4em\relax Prentice Hall International Editions, 1993.

\bibitem{CKM_tutorial}
Y.~Zeng, J.~Chen, J.~Xu, D.~Wu, X.~Xu, S.~Jin, X.~Gao, D.~Gesbert, S.~Cui, and R.~Zhang, ``A tutorial on environment-aware communications via channel knowledge map for 6g,'' \emph{IEEE Commun. Surv. Tutor.}, pp. 1--1, 2024.

\bibitem{TS38211}
\emph{Physical channels and modulation}, Std. TS 38.211 Release 15, 3GPP, Apr. 2017.

\bibitem{balanis2016antenna}
C.~A. Balanis, \emph{Antenna theory: analysis and design}.\hskip 1em plus 0.5em minus 0.4em\relax John wiley \& sons, 2016.

\bibitem{Synchr}
H.~Minn, V.~Bhargava, and K.~Letaief, ``A robust timing and frequency synchronization for ofdm systems,'' \emph{IEEE Trans. Wireless Commun.}, vol.~2, no.~4, pp. 822--839, July 2003.

\bibitem{OMP}
T.~T. Cai and L.~Wang, ``Orthogonal matching pursuit for sparse signal recovery with noise,'' \emph{IEEE Trans. Inf. Theory}, vol.~57, no.~7, pp. 4680--4688, June 2011.

\bibitem{AMP_2009}
D.~L. Donoho, A.~Maleki, and A.~Montanari, ``Message-passing algorithms for compressed sensing,'' \emph{Proc. Nat. Acad. Sci. USA}, vol. 106, no.~45, pp. 18\,914--18\,919, Nov. 2009.

\bibitem{turbo_ma}
J.~{Ma}, X.~{Yuan}, and L.~{Ping}, ``Turbo compressed sensing with partial {DFT} sensing matrix,'' \emph{IEEE Signal Process. Lett.}, vol.~22, no.~2, pp. 158--161, Feb. 2015.

\bibitem{STCS}
X.~{Kuai}, L.~{Chen}, X.~{Yuan}, and A.~{Liu}, ``Structured turbo compressed sensing for downlink massive {MIMO-OFDM} channel estimation,'' \emph{IEEE Trans. Wireless Commun.}, vol.~18, no.~8, pp. 3813--3826, Aug. 2019.

\bibitem{Jinshi}
H.~He, C.-K. Wen, S.~Jin, and G.~Y. Li, ``Deep learning-based channel estimation for beamspace mmwave massive mimo systems,'' \emph{IEEE Wireless Commun. Lett.}, vol.~7, no.~5, pp. 852--855, May 2018.

\bibitem{Vincent}
H.~Guo and V.~K.~N. Lau, ``Robust deep learning for uplink channel estimation in cellular network under inter-cell interference,'' \emph{IEEE J. Sel. Areas Commun.}, vol.~41, no.~6, pp. 1873--1887, May 2023.

\bibitem{chenchen}
C.~Liu, W.~Jiang, and X.~Yuan, ``Learning-based block-wise planar channel estimation for time-varying mimo ofdm,'' \emph{IEEE Wireless Commun. Lett.}, pp. 1--1, 2024.

\bibitem{Unify_1}
E.~Riegler, G.~E. Kirkelund, C.~N. Manch{\'o}n, M.-A. Badiu, and B.~H. Fleury, ``Merging belief propagation and the mean field approximation: A free energy approach,'' \emph{IEEE Trans. Inf. Theory}, vol.~59, no.~1, pp. 588--602, Sep. 2012.

\bibitem{Unify_2}
D.~Zhang, X.~Song, W.~Wang, G.~Fettweis, and X.~Gao, ``Unifying message passing algorithms under the framework of constrained bethe free energy minimization,'' \emph{IEEE Trans. Wireless Commun.}, vol.~20, no.~7, pp. 4144--4158, Sep. 2021.

\bibitem{HVMP}
H.~Jiang, X.~Yuan, and Q.~Guo, ``Generalized bilinear factorization via hybrid vector message passing,'' \emph{arXiv preprint arXiv:2401.03626}, 2024.

\bibitem{VMP}
J.~Winn, C.~M. Bishop, and T.~Jaakkola, ``Variational message passing.'' \emph{J. Mach. Learn. Res.}, vol.~6, no.~4, Apr. 2005.

\bibitem{EP}
T.~P. Minka, ``Expectation propagation for approximate bayesian inference,'' in \emph{Proc. Conf. Uncertainty Artif. Intell.}, ser. UAI'01, San Francisco, CA, USA, 2001, p. 362–369.

\bibitem{bishop2006pattern}
C.~M. Bishop, ``Pattern recognition and machine learning,'' \emph{Springer}, vol.~2, pp. 1122--1128, 2006.

\bibitem{guiasu1985principle}
S.~Guiasu and A.~Shenitzer, ``The principle of maximum entropy,'' \emph{The mathematical intelligencer}, vol.~7, pp. 42--48, 1985.

\bibitem{mardia2009directional}
K.~V. Mardia and P.~E. Jupp, \emph{Directional statistics}.\hskip 1em plus 0.5em minus 0.4em\relax John Wiley \& Sons, 2009.

\bibitem{VALSE}
M.-A. Badiu, T.~L. Hansen, and B.~H. Fleury, ``Variational {Bayesian} inference of line spectra,'' \emph{IEEE Trans. Signal Process.}, vol.~65, no.~9, pp. 2247--2261, Jan. 2017.

\bibitem{IRC}
F.~M.~L. Tavares, G.~Berardinelli, N.~H. Mahmood, T.~B. Sorensen, and P.~Mogensen, ``On the impact of receiver imperfections on the mmse-irc receiver performance in {5G} networks,'' in \emph{IEEE 79th Veh. Tech. Conf. (VTC Spring)}, 2014, pp. 1--6.

\end{thebibliography}

\end{document}